\documentclass[acmsmall]{acmart}

\renewcommand\footnotetext[1]{} 
\renewcommand\footnotetextauthorsaddresses[1]{} 
\renewcommand\footnotetextcopyrightpermission[1]{} 
\fancypagestyle{firstpagestyle}{
  \fancyhf{}
}

\setcopyright{none}

\setcitestyle{nosort}

\usepackage[author=anonymous,nomargin,marginclue,footnote,status=final]{fixme}
\FXRegisterAuthor{ls}{als}{LS}
\FXRegisterAuthor{dp}{atm}{dp}
\FXRegisterAuthor{bk}{abk}{bk}
\FXRegisterAuthor{pw}{apw}{PW}

\usepackage{booktabs}   
\usepackage{subcaption} 

\usepackage{mathtools}

\usepackage[all]{xy}

\newcommand{\ST}{\mathsf{ST}}

\newcommand{\Rat}{\mathbb{Q}}
\newcommand{\RatI}{\Rat\cap[0,1]}

\newcommand{\At}{\mathsf{At}}
\newcommand{\Act}{\mathsf{Act}}

\newcommand{\ball}[3]{B_{#2}({#3})}

\newcommand{\id}{\mathsf{id}}

\newcommand{\Set}{\mathsf{Set}}

\theoremstyle{plain}
\newtheorem{thm}{Theorem}[section]
\newtheorem{lem}[thm]{Lemma}

\theoremstyle{definition}
\newtheorem{defn}[thm]{Definition}
\newtheorem{expl}[thm]{Example}

\newtheorem{rem}[thm]{Remark}

\newcommand{\nonexp}[2]{#1 \to_1 #2}
\newcommand{\supnorm}[1]{\lVert #1 \rVert_\infty}
\newcommand{\nbhood}[2]{U^{#1}(#2)}
\newcommand{\CA}{\mathcal{A}}
\newcommand{\CB}{\mathcal{B}}
\newcommand{\CC}{\mathcal{C}}
\newcommand{\CF}{\mathcal{F}}
\newcommand{\qr}{\mathsf{qr}}
\newcommand{\rk}{\mathsf{rk}}
\newcommand{\dif}{\,\mathrm{d}}
\newcommand{\intsuc}[2]{\textstyle{\int} #2 \dif\pi_{#1}}
\newcommand{\dfun}{\mathcal{D}}

\newcommand{\modf}[1]{\mathcal{L}_{#1}}
\newcommand{\diabind}[3]{#1 \Diamond\lceil #2: #3\rceil}
\newcommand{\last}{\mathsf{last}}

\begin{document}

\title[A Probabilistic van Benthem Theorem]{A van Benthem Theorem\\ for Quantitative Probabilistic Modal Logic}         


\author{Paul Wild}
\affiliation{
  \institution{Friedrich-Alexander-Universit\"at Erlangen-N\"urnberg}
  \country{Germany}
}
\email{paul.wild@fau.de}

\author{Lutz Schr\"oder}
\affiliation{
  \institution{Friedrich-Alexander-Universit\"at Erlangen-N\"urnberg}
  \country{Germany}
}
\email{lutz.schroeder@fau.de}

\author{Dirk Pattinson}
\affiliation{
  \institution{Australian National University, Canberra}
  \country{Australia}
}
\email{dirk.pattinson@anu.edu.au}
\author{Barbara K\"onig}
\affiliation{
  \institution{Universit\"at Duisburg-Essen}
  \country{Germany}
}
\email{barbara_koenig@uni-due.de}

\begin{abstract} In probabilistic transition systems, behavioural metrics provide a more fine-grained and stable measure of system equivalence than crisp notions of bisimilarity. They correlate strongly to quantitative probabilistic logics, and in fact the distance induced by a probabilistic modal logic taking values in the real unit interval has been shown to coincide with behavioural distance. For probabilistic systems, probabilistic modal logic thus plays an analogous role to that of Hennessy-Milner logic on classical labelled transition systems. In the quantitative setting, invariance of modal logic under bisimilarity becomes non-expansivity of formula evaluation w.r.t.\ behavioural distance. In the present paper, we provide a characterization of the expressive power of probabilistic modal logic based on this observation: We prove a probabilistic analogue of the classical van Benthem theorem, which states that modal logic is precisely the bisimulation-invariant fragment of first-order logic. Specifically, we show that quantitative probabilistic modal logic lies dense in the bisimulation-invariant fragment, in the indicated sense of non-expansive formula evaluation, of quantitative probabilistic first-order logic; more precisely, bisimulation-invariant first-order formulas are approximable by modal formulas of bounded rank.

For a description logic perspective on the same result,
  see~\cite{wspk:modal-charact-prob-fuzzy-arxiv}.
\end{abstract}

\ccsdesc[500]{Software and its engineering~General programming languages}
\ccsdesc[300]{Social and professional topics~History of programming languages}


\maketitle

\section{Introduction}

Probabilistic transition systems, or partial Markov chains, serve as a
quantitative model of concurrent systems (see~\cite{vanGlabbeekEA95}
for an overview of probabilistic models of concurrency). Probabilistic
systems can be compared under standard two-valued (\emph{crisp})
notions of bisimilarity~\cite{LarsenSkou89,BluteEA97} under which two
states are either bisimilar or not, but it has been observed
previously~\cite{GiacaloneEA90} that in many respects, quantitative
measures of process equivalence are more suitable in this setting:
Probabilistic systems may, e.g., differ slightly in the values of
individual probabilities or contain mutually deviating but very
unlikely transitions, and in such cases one would like to have the
possibility of saying that two processes are almost the same, or in
fact quantifying their degree of distinctness. This has motivated the
introduction of \emph{behavioural metrics} measuring the
\emph{behavioural distance} between states in probabilistic
systems~\cite{GiacaloneEA90,DesharnaisEA99,BreugelWorrell05,DesharnaisEA08,bbkk:behavioral-metrics-functor-lifting,CastiglioniEA16}. More
precisely, these distance functions are \emph{pseudometrics}, i.e.\
distinct states can have distance zero, namely if they are exactly
bisimilar.

From the outset, both crisp probabilistic bisimilarity and behavioural
metrics have been related to suitable modal logics. Larsen and
Skou~\cite{LarsenSkou89} introduce a modal logic featuring modalities
$\Diamond_p$, with $\Diamond_p\phi$ read `with probability at
least~$p$, the state reached in the next step will
satisfy~$\phi$'. This logic thus has a two-valued semantics, and we
refer to it as \emph{crisp probabilistic modal logic}. It is easy to
see that this logic is \emph{bisimulation-invariant}, i.e.\ if two
states are probabilistically bisimilar then they satisfy the same
modal formulas; Larsen and Skou show that the converse holds as well
under additional assumptions on the underlying systems, which amounts
to a probabilistic Hennessy-Milner theorem. In the more fine-grained
setting of behavioural metrics, bisimulation invariance becomes
non-expansivity w.r.t.\ behavioural distance -- to see the connection,
consider crisp bisimilarity as a $\{0,1\}$-valued discrete behavioural
distance and observe that in this view, a map from the state space
into a set $\{0,1\}$ of crisp truth values is bisimulation-invariant
iff it is non-expansive w.r.t.\ this distance. Van Breugel and
Worrell~\cite{BreugelWorrell05} introduce a behavioural
\mbox{(pseudo-)}metric on probabilistic transition systems
(generalizing previous work), along with a \emph{quantitative
  probabilistic modal logic} taking values in the unit interval
(closely related to logics previously introduced by
Kozen~\cite{Kozen85} and Desharnais et al.~\cite{DesharnaisEA99})
whose key feature is a modality taking expected truth values. This
logic, a fragment of the probabilistic
$\mu$-calculus~\cite{CleavelandEA05,HuthKwiatkowska97}, is
\emph{bisimulation-invariant} in the mentioned sense, i.e.\
non-expansive w.r.t.\ behavioural distance. Moreover, the pseudometric
on states induced by the logic in fact coincides (up to a constant
factor) with the behavioural metric, a quantitative version of the
Hennessy-Milner theorem.

Classically, i.e.\ for modal logics of relational structures, one has
a second converse to bisimulation invariance of modal logic besides
the Hennessy-Milner theorem: The \emph{van Benthem
  theorem}~\cite{BenthemThesis} asserts that every
bisimulation-invariant first-order property is in fact expressible in
modal logic -- i.e.\ modal logic is as expressive as it can be, given
that it embeds into first-order logic and is
bisimulation-invariant. This result can be viewed as saying that modal
logic provides effective syntax for bisimulation-invariant first-order
properties -- something that first-order logic itself does not, as it
is undecidable whether a given first-order formula is
bisimulation-invariant~\cite{Otto06}.

Our main result in the present paper is a corresponding expressive
completeness result for quantitative probabilistic modal logic. Before
we discuss the statement of this result in more detail, we recall two
earlier results for logics that are related, in orthogonal dimensions,
to our target logic:
\begin{itemize}
\item The only currently known expressive completeness result for
  \emph{crisp} probabilistic modal logic states that every
  bisimulation-invariant property expressible in a natural variant of
  (crisp) probabilistic first-order logic is expressible in crisp
  probabilistic modal logic \emph{extended with infinite conjunction}
  (hence also infinite disjunction) by a formula of \emph{bounded
    modal rank}. (This is an instance of a general result established
  in coalgebraic
  logic~\cite{SchroderPattinson10b,LitakEA12,SchroderEA17}.)
\item A recently established expressive completeness result for a
  simple quantitative logic, \emph{fuzzy modal logic} with Zadeh
  semantics of the propositional connectives, states that every
  property that is bisimulation-invariant in the sense of being
  non-expansive w.r.t.\ the natural notion of behavioural distance and
  moreover expressible in fuzzy first-order logic can be
  \emph{approximated} by formulas in fuzzy modal logic \emph{of
    bounded modal rank}~\cite{WildEA18}.
\end{itemize}
One sees an apparent analogy between infinite conjunctions and
approximation. The bound on the rank is essential in two senses:
First, without it, the statement becomes, in both cases, morally a
direct consequence of the much simpler Hennessy-Milner theorem, and
then in fact applies to arbitrary bisimulation-invariant properties
rather than only first-order definable ones. This is already true in
the case of relational labelled transition systems: By the standard
Hennessy-Milner theorem, \emph{every} bisimulation-invariant property,
first-order definable or not, is definable in Hennessy-Milner logic
with infinite conjunction. (This follows simply from the fact that a
bisimulation-invariant property is a union of bisimilarity equivalence
classes, and by the Hennessy-Milner theorem each such class can be
described by the infinite conjunction of all modal formulas satisfied
by the states in the class.) An analogous statement holds for crisp
probabilistic modal logic (by Larsen and Skou's probabilistic variant
of the Hennessy-Milner theorem~\cite{LarsenSkou89}); for quantitative
probabilistic modal logic, van Breugel and
Worrell~\cite{BreugelWorrell05} similarly show that \emph{every}
bisimulation-invariant property in the metric sense can be
approximated by modal formulas, again as a direct consequence of their
quantitative analogue of the Hennessy-Milner theorem. Second, in the
classical case of labelled transition systems (or Kripke models), the
bound on the rank is actually the core of the van Benthem theorem:
Once one knows that every bisimulation-invariant first-order definable
property is definable by a modal formula with infinite conjunction but
of bounded rank, the actual van Benthem theorem, i.e.\ definability by
a \emph{finitary} modal formula, is immediate since (under the
standard assumption that the modal language has only finitely many
atoms, which is w.l.o.g.\ and made in all existing proofs of the
theorem) there are, up to logical equivalence, only finitely many
formulas of a given bounded rank. Summing up, the bound on the rank is
the key part of the van Benthem theorem.

Correspondingly, our main result, i.e.\ the announced expressive
completeness result for quantitative probabilistic modal logic, states
that
\begin{quotation}
  \emph{every bisimulation-invariant property that is definable in
    quantitative probabilistic first-order logic is approximable by
    quantitative probabilistic modal formulas of bounded rank}.
\end{quotation}
Again, bisimulation invariance is to be understood in the sense of
non-expansiveness w.r.t.\ behavioural distance. Quantitative
probabilistic first-order logic is a natural first-order extension of
quantitative probabilistic modal logic that we introduce here; its
syntax and semantics are modelled on coalgebraic predicate
logic~\cite{LitakEA12} and ultimately Chang's modal predicate
logic~\cite{Chang73}, but we replace the original two-valued notion of
satisfaction with a quantitative notion. Simultaneously, quantitative
probabilistic modal logic can be seen as a quantitative variant of
Halpern's (crisp) \emph{type-$1$} (or \emph{statistical})
probabilistic first-order logic~\cite{Halpern90}.

Technically, we base our proof on a strategy put forward by
Otto~\cite{o:van-Benthem-Rosen-elementary} and used in many recent
van-Benthem type results (including~\cite{SchroderEA17,WildEA18}): One
shows using a suitable notion of Ehrenfeucht-Fra\"iss\'e equivalence
that every bisimulation-invariant first-order property is
\emph{local}, i.e.\ depends only on a bounded neighbourhood under an
adapted notion of Gaifman distance, then concludes by an unravelling
construction that the target property is in fact invariant under
bounded-depth bisimilarity, and finally shows that all such invariant
properties are approximable by modal formulas of bounded depth. Unlike
in the classical case, where the statement that all properties that
are invariant under $k$-bounded bisimulation are definable by a modal
formula of modal rank~$k$ is next to trivial, the last step is in fact
the key part of this programme in the quantitative setting, and
presumably of independent interest. In particular, modal
approximability in bounded rank holds for bounded-depth
bisimulation-invariant properties irrespective of their first-order
definability, a statement that certainly cannot be improved to
on-the-nose modal definability. 

Besides newly introduced notions of Ehrenfeucht-Fra\"iss\'e
equivalence and Gaifman distance on probabilistic transition systems,
the key tool in the proof is a notion of up-to-$\epsilon$ bisimulation
game for probabilistic transition systems. The pseudometric induced by
this game coincides with logical distance and (hence) with behavioural
distance in the sense of van Breugel and
Worrell~\cite{BreugelWorrell05}. The proof of this equivalence is
based on Kantorovich-Rubinstein duality (see
also~\cite{bbkk:behavioral-metrics-functor-lifting,BreugelEA08,Breugel17});
essentially, the logical distance relates to Kantorovich distance, and
the game distance to Wasserstein distance, and Kantorovich-Rubinstein
duality guarantees that these distances are equal. Our games differ
substantially both from the bisimulation games up to~$\epsilon$
previously considered by Desharnais et al.~\cite{DesharnaisEA08} and
from the much simpler games used by Wild et al.~\cite{WildEA18} for
the case of fuzzy relational systems in several respects; in
particular, in our game the allowed deviation~$\epsilon$ changes in
the course of a play. Also, the behavioural distance induced by our
games is incomparable to that induced by that of Desharnais et
al. (who already show that their distance is incomparable to van
Breugel and Worrell's~\cite[Examples~7 and~8]{DesharnaisEA08}, and
relates more closely to logical distances induced by \emph{crisp}
probabilistic modal logics).

The material is organized as follows. We introduce the relevant logics
in Section~\ref{sec:logics}; that is, we recall the definition of
quantitative probabilistic modal logic and newly introduce its
first-order extension. In Section~\ref{sec:games}, we introduce
various notions of behavioural distance, including our notion of
bisimulation game up to~$\epsilon$ as well as Kantorovich and
Wasserstein distances, both based on fixed point definitions. As
indicated above, the key stepping stone towards our probabilistic van
Benthem theorem is the result that \emph{all} properties that are
invariant under bounded depth bisimulation are approximable by modal
formulas of bounded rank, proved using Kantorovich-Rubinstein duality
in Section~\ref{sec:modal-approx}. We introduce our notion of
Ehrenfeucht-Fra\"iss\'e equivalence and Gaifman distance and
subsequently prove locality of bisimulation-invariant quantitative
probabilistic first-order formulas in
Section~\ref{sec:locality}. Finally, we prove our main result, the van
Benthem theorem for quantitative probabilistic modal logic, in
Section~\ref{sec:main}.

\paragraph{Related Work} Our work owes much to Wild et al.'s recent
formulation and proof of a van Benthem theorem for fuzzy modal logic
as cited above~\cite{WildEA18}. This result is set in a much simpler
framework where the logic is interpreted over fuzzy relational models,
which differ from classical labelled transition systems by assigning
truth degrees in $[0,1]$ to propositions and transitions but unlike
probabilistic models do not impose restrictions on the sum of truth
degrees. Also, modalities are interpreted in the fuzzy setting by just
taking infima and suprema, respectively, rather than expected truth
values as in (quantitative) probabilistic modal logic. Summing up,
quantitative probabilistic modal logic is semantically more complex
than fuzzy modal logic in a) involving real arithmetic as opposed to
just lattice operations, and b) consequently not allowing for a
separate treatment of successors, precisely because it involves
summation over sets of successors. Technically, this is reflected
mainly in the more complicated structure of behavioural metrics and
bisimulation games; in particular, unlike in the fuzzy relational case
we need to include a Wasserstein formulation of behavioural distance.

As shown by Rosen~\cite{Rosen97}, the classical van Benthem theorem
holds also over finite structures; although we use a proof strategy
that covers the finite case in the classical setting, we currently
leave open the question whether our main result remains true over
finite probabilistic transition systems (essentially, constructions
that produce finite structures in the classical case become infinite
in the presence of infinitely many truth values).

For two-valued logic, van-Benthem-type theorems, also known as
\emph{modal characterization theorems}, abound, having been
established e.g.\ for logics with frame conditions~\cite{DawarOtto05},
neighbourhood logic~\cite{HansenEA09}, fragments of
XPath~\cite{tenCateEA10,FigueiraEA15,AbriolaEA17}, modal $\mu$-calculi
(within monadic second order
logics)~\cite{JaninWalukiewicz95,EnqvistEA15}, PDL (within weak chain
logic)~\cite{Carreiro15}, modal first-order
logics~\cite{Benthem01,SturmWolter01} (within first-order
correspondence languages), and two-dimensional modal logics with an
$S5$-modality~\cite{WildSchroder17} (within $S5$ modal first-order
logic). We are not aware of previous modal characterization theorems
in the quantitative setting other than the mentioned work on the fuzzy
case~\cite{WildEA18}.

We have already mentioned work on behavioural distance in
probabilistic
systems~\cite{GiacaloneEA90,DesharnaisEA99,BreugelWorrell05,DesharnaisEA08,bbkk:behavioral-metrics-functor-lifting,CastiglioniEA16};
concretely, the Kantorovich-style discount-free notion of behavioural
distance that we use here goes back to work by van Breugel et
al.~\cite{BreugelEA07}.

Our probabilistic Ehrenfeucht-Fra\"iss\'e games (but not our
bisimulation games) are related to corresponding games introduced by
Makowski and Ziegler in the context of topological first-order
logic~\cite{MakowskyZiegler80} as well as to probabilistic
bisimulation games used by Desharnais et al.~\cite{DesharnaisEA08}, in
that they include rounds where sets of states are played in
intermediate configurations, with the main difference being that our
games involve fuzzy subsets, rather than crisp ones as in the cited
work.

\section{Quantitative Probabilistic Logics}\label{sec:logics}

We proceed to introduce the logics featuring in our main result,
quantitative probabilistic modal logic and quantitative probabilistic
first-order logic. Both logics are interpreted over
\emph{probabilistic transition systems}, which we often refer to just
as \emph{models}. We allow for infinite transition systems but, like
existing work on quantitative probabilistic modal
logic~\cite{BreugelWorrell05}, restrict to discrete probability
distributions over successors at each state.

Explicitly, we fix a set~$\At$ of \emph{(propositional) atoms}; then a
\emph{probabilistic transition system}
\begin{equation*}
\CA = (A,(p^\CA)_{p\in\At},\pi^\CA)
\end{equation*}
consists of a set $A$ of \emph{states}, a valuation map
$p^\CA\colon A\to[0,1]$ for every atom $p$, and a map
$\pi^\CA\colon A\times A\to[0,1]$ (we will write $\pi$ instead of
$\pi^\CA$ when the model is clear from the context) such that for each
$a\in A$, the map
\[ \pi_a\colon A\to[0,1], \quad \pi_a(a') = \pi(a,a') \]
is either zero or is a discrete probability measure on $A$, i.e.
\[\sum_{a'\in A}\pi_a(a') \in \{0,1\}\]
(implying in the latter case that the \emph{support}
$\{a'\in A\mid \pi_a(a')>0\}$ of $\pi_a$ is at most countable). We call
a state~$a$ \emph{terminating} if $\sum_{a'\in A}\pi_a(a') = 0$, and
\emph{transient} otherwise. At transient states, $\pi$ thus acts as a
probabilistic transition relation.  Whenever a model is designated by
a calligraphic letter, such as~$\CA$, we will always implicitly
assume that the corresponding roman letter, e.g.~$A$, designates the
set of states.

\begin{rem}\label{rem:systems}
  The probabilistic transition systems that we define above can be
  seen as Markov chains extended with propositional atoms and the
  possibility of termination. They deviate from the ones considered in
  previous work on quantitative probabilistic modal
  logic~\cite{BreugelWorrell05} in three mostly inessential ways: a)
  We consider only one probabilistic transition relation. This is
  purely in the interest of readability; a generalization to several
  probabilistic transition relations indexed over a set of actions
  amounts to no more than adding more indices. b) For the sake of
  generality, we have added propositional atoms; the set of
  propositional atoms is a parameter of the setup, so the model
  without atoms is a special case. c) Instead of using
  subdistributions over successor states, i.e.\ requiring
  $\sum_{a'\in A}\pi_a(a') \le 1$, we more specifically require either
  a distribution (total weight~$1$) or termination (total weight~$0$)
  at each state. This is for technical convenience and clarity in
  presenting the proofs, in particular the bisimulation game; minor
  technical modifications required to cover also the model based on
  unrestricted subdistributions are summarized in
  Remark~\ref{rem:subdistributions}.
\end{rem}

\subsection{Quantitative Probabilistic Modal Logic}
\noindent We next recall the syntax and semantics of quantitative
probabilistic modal logic, following van Breugel and
Worrell~\cite{BreugelWorrell05}. \emph{Formulas} $\phi,\psi,\dots$ of
the logic are defined by the grammar
\begin{equation*}
  \phi,\psi::= c\mid p\mid\phi\ominus c\mid \neg\phi\mid\phi\land\psi
  \mid \Diamond \phi
\end{equation*}
where $c\in\RatI$, and $p\in\At$ ranges over propositional
atoms. Maybe slightly deviating from standard practice, we define the
\emph{(modal) rank} of a modal formula~$\phi$ as the maximal nesting
depth of~$\Diamond$ \emph{and propositional atoms} in~$\phi$; e.g.\
$\Diamond\Diamond p\land\Diamond q$ has rank~$3$ (since~$p$
contributes~$1$ to the rank). We denote the rank of a formula $\phi$
by~$\rk\phi$ and the set of all formulas of rank at most~$n$
by~$\modf{n}$.

A formula~$\phi$ evaluates to a probabilistic truth value
\begin{equation*}
  \phi(a)\in[0,1]
\end{equation*}
at a state~$a$ in a probabilistic transition system~$\CA$. Conjunction
is interpreted by taking minima and negation by taking
complementary probability, while $\ominus$ is subtraction truncated
at~$0$. The modal operator~$\Diamond$ takes expected truth
values. Formally, we define $\phi(a)$ recursively by
\begin{align*}
  c(a) & = c \\ 
  p(a) & = p^\CA(a) \\(\phi\ominus c)(a) & = \max(\phi(a)-c,0) \\
  (\neg\phi)(a) &=1-\phi(a) \\ (\phi\land\psi)(a) & = \min(\phi(a),\psi(a))\\
  (\Diamond \phi)(a) & \textstyle=\intsuc{a}{\phi}.
\end{align*}
Note that we generally use integral notation for readability, although
given that all distributions are discrete, these integrals are
actually just infinite sums; e.g., in the above case,
\begin{equation*}
  \intsuc{a}{\phi}=\sum_{a'\in A}
  \phi(a')\cdot\pi_a(a')=\sum_{a'\in A}
  \phi(a')\cdot\pi(a,a')
\end{equation*}
is the expected truth value of~$\phi$ for a random successor of~$a$,
distributed according to~$\pi_a$. We define disjunction~$\lor$ as the
dual of~$\land$ as usual, so that~$\lor$ takes maxima.
\begin{rem}\label{rem:box-examples}
  We note that the dual~$\Box$ of~$\Diamond$ defined by
  $\Box\phi=\neg\Diamond\neg\phi$ differs from $\Diamond$ only at
  terminating states: For~$a$ terminating, we have
  $(\Diamond\phi)(a)=0$ for all~$\phi$ and hence $(\Box\phi)(a)=1$ for
  all~$\phi$, while for~$a$ transient, we have
  \begin{equation*}
    (\Box\phi)(a)=1-\intsuc{a}{(1-\phi)}=1-\intsuc{a}{1}+\intsuc{a}{\phi}=\intsuc{a}{\phi}=(\Diamond\phi)(a).
  \end{equation*}
  For instance, the formula $\Diamond\Box 0$ gives the probability of
  reaching a terminating state (note that the expected value of a
  $\{0,1\}$-valued function~$f$ is just the probability of
  $f^{-1}[\{1\}]$), so
  \begin{equation*}
    \Diamond\Diamond\Box 0
  \end{equation*}
  gives the expected value, taken over successor states in the first
  step, of the probability of reaching, in the second step, a
  terminating state.
\end{rem}
\begin{rem}
  The probabilistic
  $\mu$-calculus~\cite{CleavelandEA05,HuthKwiatkowska97} extends the
  above syntax by adding fixed point operators and moreover interprets
  conjunction as \L{}ukasiewicz fuzzy conjunction~$\otimes$, given by
  $r\otimes q=\max(r+q-1,0)$, instead of as minimum. The latter
  interpretation (as minimum) is referred to as \emph{Zadeh semantics}. Zadeh
  semantics is well-known to embed into \L{}ukasiewicz semantics by a
  simple translation, so as indicated above quantitative probabilistic
  modal logic (without~$\ominus$) is a fragment of the probabilistic
  $\mu$-calculus. The main reason that Zadeh instead of \L{}ukasiewicz
  semantics is used in the present work and also by van Breugel and
  Worrell~\cite{BreugelWorrell05} is that \L{}ukasiewicz conjunction
  would fail to be non-expansive (e.g.\ $r\otimes r=\max(2r-1,0)$).  In
  fact, as already pointed out by Wild et al.~\cite{WildEA18} it would
  make logical distance discrete as it allows for arbitrary
  amplification of small deviations of truth degrees. Additional
  difficulties would arise with Kantorovich distance.
\end{rem}
\subsection{Quantitative Probabilistic First-Order Logic}

\noindent As the first-order correspondence language of quantitative
probabilistic modal logic, we now proceed to introduce
\emph{quantitative probabilistic first-order logic}, with
\emph{formulas} $\phi,\psi,\dots$ defined by the grammar 
\begin{equation*}
  \phi,\psi::= c\mid p(x)\mid x=y\mid\phi\ominus c\mid \neg\phi\mid\phi\land\psi
  \mid \exists x.\,\phi\mid \diabind{x}{y}{\phi}.
\end{equation*}
Again, $p$ ranges over propositional atoms and~$c$ over $\RatI$
while~$x$ and~$y$ range over a fixed countably infinite reservoir of
\emph{variables}. The semantics of the propositional part is
essentially as in the modal logic. Equality is crisp. Existential
quantification is interpreted by taking suprema, and formulas of the
form $\diabind{x}{y}{\phi}$ denote the expected truth value of~$\phi$
at successors~$y$ of~$x$. We have the expected notions of free and
bound variables, under the additional proviso that~$y$ (but not~$x$!)
is bound in $\diabind{x}{y}{\phi}$. The \emph{(quantifier) rank} of a
formula~$\phi$ is the maximal nesting depth of the variable-binding
operators~$\exists$ and~$\Diamond$ and propositional atoms~$p$
in~$\phi$; e.g.~$\exists x.\,\diabind{x}{y}{p(y)}$ has rank~$3$. We
denote the quantifier rank of a probabilistic first-order formula
$\phi$ by $\qr(\phi)$.

Formally, we define the semantics of the logic by assigning a truth
value $\phi(\bar a)\in[0,1]$ to a formula $\phi(x_1,\dots,x_n)$ with
free variables at most $x_1,\dots,x_n$, depending on a probabilistic
transition system $\CA = (A,(p^\CA)_{p\in\At},\pi^\CA)$ and a vector
$\bar a=(a_1,\dots,a_n)\in A^n$ of values for the free variables. We
define $\phi(\bar a)$ recursively by essentially the same clauses as
in quantitative probabilistic modal logic for the propositional
constructs, and
\begin{align*}
  p(x_i)(\bar a) & = p^\CA(a_i)\\
  (x_i=x_j)(\bar a) & = 1\text{ if $a_i=a_j$, and $0$ otherwise}\\
  (\exists x_0.\,\phi(x_0,x_1,\dots,x_n))(\bar a) 
                 &  = \textstyle \bigvee_{a_0\in A}\phi(a_0,a_1,\dots, a_n)\\
  (\diabind{x_i}{y}{\phi(y,x_1,\dots,x_n)})(\bar a) 
                 & = \intsuc{a_i}{\phi(\,\cdot\,,a_1,\dots,a_n)}.
\end{align*}
\begin{expl}
  In quantitative probabilistic first-order logic, we can express the
  transition probability from~$x$ to~$y$ as $\diabind{x}{z}{z=y}$, and
  the probability of a finite set $\{y_1,\dots,y_n\}$ as
  $\diabind{x}{z}{z=y_1\lor\dots\lor z=y_n}$. The formula
  $\phi=\diabind{x}{z}{\diabind{z}{w}{w=y}}$ denotes the expected
  probability, in the next step, of reaching~$y$ after another step,
  which indeed coincides with the more intuitive reading of~$\phi$ as
  the probability of reaching~$y$ from~$x$ in two independently
  distributed steps. The formula $\exists y.\,\diabind{x}{z}{z=y}$
  denotes, roughly, the probability of the most probable successor
  of~$x$, or more precisely the supremum over the probabilities of all
  successors.
\end{expl}
\noindent We have a \emph{standard translation}~$\ST$ from
quantitative probabilistic modal logic into quantitative probabilistic
first-order logic. As in the classical case, $\ST$ is indexed over a
variable~$x$ representing the current evaluation point. For a modal
formula~$\phi$, we define $\ST_x(\phi)$ recursively by
\begin{align*}
  \ST_x(p) & = p(x) \\
  \ST_x(\Diamond\phi) & = \diabind{x}{y}{\ST_y(\phi)},
\end{align*}
and by commutation with all other constructs. An easy induction shows
that the standard translation preserves probabilistic truth degrees:
\begin{lem}
  For every modal formula~$\phi$ and every state~$a$ in a
  probabilistic transition system, $\phi(a)=\ST_x(\phi)(a)$.
\end{lem}
\noindent The standard translation thus identifies quantitative
probabilistic modal logic as a fragment of quantitative probabilistic
first-order logic.
\begin{rem}
  As indicated in the introduction, we take the treatment of the
  modality~$\Diamond$ in our first-order syntax from Litak et al.'s
  (two-valued) \emph{coalgebraic predicate logic}~\cite{LitakEA12},
  which in turn is inspired by Chang's \emph{modal predicate
    logic}~\cite{Chang73}; also, the syntactic definition of the above
  standard translation essentially coincides with the translation from
  coalgebraic modal logic into coalgebraic predicate logic. Litak et
  al.~\cite{LitakEA13} point out that coalgebraic predicate logic may
  equivalently be seen as the extension of the purely modal logic to a
  hybrid logic with nominals, satisfaction operators, local binding,
  and the universal modality. The same comment applies to our
  quantitative logic.

  If one were to use our expectation operator $\Diamond$ in the
  context of a two-valued logic, then expectation would just turn into
  probability (we have already observed in
  Remark~\ref{rem:box-examples} that the expected value of a
  $\{0,1\}$-valued function~$f$ is just the probability
  of~$f^{-1}[\{1\}]$), and moreover one would then have to convert
  probabilities into binary truth values, say by using them in
  arithmetic conditions. This is precisely what happens in Halpern's
  two-valued \emph{type-1} or \emph{statistical} probabilistic
  first-order logic, which features weight expressions $w_x(\phi)$
  that denote the probability of a randomly sampled state~$x$ to
  satisfy~$\phi$, and are used in formulas of first-order real
  arithmetic~\cite{Halpern90}. (More precisely, the weight operator
  can more generally be applied to vectors of variables; we leave a
  corresponding extension of our logic to future work.) Semantically,
  type-1 probabilistic first-order logic further differs from our
  above logic in that it is interpreted over structures that have
  crisp predicates and use only a single global probability
  distribution on the state set, instead of one distribution per
  state. In our syntax, a global distribution can be emulated by
  restricting the $\diabind{x}{y}{\phi}$ construct to be applied only
  to a single fixed globally free variable~$x$. Summing up,
  quantitative probabilistic first-order logic can, as suggested
  earlier, be seen as a quantitative variant of type-1 probabilistic
  first-order logic.
\end{rem}
\subsection{Coalgebraic Modelling}\label{sec:coalg}
Universal coalgebra~\cite{Rutten00} serves as a generic framework for
modelling state-based systems, with the system type encapsulated as a
set functor. Although we are only concerned with a concrete system
type, viz.\ probabilistic transition systems, in the present paper, we
do need coalgebraic methods to some degree. In particular, the
requisite background on behavioural
distances~\cite{BreugelWorrell05,BreugelEA08,bbkk:behavioral-metrics-functor-lifting}
is largely based on coalgebraic techniques, and moreover we will need
the final coalgebra at one point in the development. We require only
basic definitions, which we recapitulate here and then instantiate to
the case of probabilistic transition systems.

Recall first that a set functor~$F:\Set\to\Set$ consists of an
assignment of a set~$FX$ to every set~$X$ and a map $Ff:FX\to FY$ to
every map $f:X\to Y$, preserving identities and composition. The core
example of a functor for the present purposes is the
\emph{distribution functor}~$\dfun$, which assigns to a set $X$ the
set $\dfun X$ of discrete probability measures on~$X$, and to a map
$f:X\to Y$ the map $\dfun f:\dfun X\to\dfun Y$ that takes image
measures; explicitly, $\dfun f(\mu)$ is the image measure of~$\mu$
along~$f$, given by $\dfun f(\mu)(A)=\mu(f^{-1}[A])$. Functors can be
combined by taking \emph{products} and \emph{sums}: Given set functors
$F,G:\Set\to\Set$, the set functors $F\times G,F+G:\Set\to\Set$ are
given by $(F\times G)X=FX\times GX$ and $(F+G)X=FX+GX$, respectively,
with the evident action on maps in both cases; here, $+$ denotes
disjoint union as usual. Every set~$C$ induces a \emph{constant
  functor}, also denoted~$C$ and given by $CX=C$ and $Cf=\id_C$ for
every set~$X$ and every map~$f$. Moreover, the \emph{identity
  functor}~$\id$ is given by $\id\, X=X$ and $\id\, f=f$ for all sets~$X$
and all maps~$f$.

An \emph{$F$-coalgebra} $(A,\xi)$ for a set functor~$F$ consists of a
set~$X$ of \emph{states} and a \emph{transition map} $\xi:A\to FA$,
thought of as assigning to each state $a\in A$ a structured collection
$\xi(a)$ of successors.  A $\dfun$-coalgebra $(A,\xi)$, for instance, is just a
Markov chain: its transition map $\xi:A\to\dfun A$ assigns to each
state a distribution over successor states. Similarly, probabilistic
transition systems in the sense defined above are coalgebras $(A,\xi)$
for the set functor $[0,1]^\At\times(\dfun+1)$: If $\xi(a)=(f,\pi)$,
then $f:\At\to[0,1]$ determines the truth values of the propositional
atoms at the state~$a$, and $\pi$ is either a discrete probability
measure determining the successors of~$a$ or a designated value
denoting termination. The variant of probabilistic transition systems
considered by van Breugel and Worrell as discussed in
Remark~\ref{rem:systems}, which indexes probabilistic transition
relations over a set~$\Act$ of actions and moreover uses unrestricted
subdistributions, corresponds to coalgebras $(A,\xi)$ for the set
functor $\dfun(\id+1)^\Act$ -- given a state~$a$ and an action
$c\in\Act$, $\xi(a)(c)\in\dfun(A+1)$ is a subdistribution over
successor states of~$a$, with the summand~$1$ serving to absorb the
weight missing to obtain total weight~$1$.

A \emph{morphism} $f:(A,\xi)\to(B,\zeta)$ between $F$-coalgebras
$(A,\xi)$ and $(B,\zeta)$ is a map $f:A\to B$ such that
\begin{equation*}
  Ff(\xi(a)) = \zeta(f(a)) 
\end{equation*}
for all states $a\in A$. Morphisms should be thought of as
behaviour-preserving maps or functional bisimulations. E.g.\
$f:A\to B$ is a morphism of $\dfun$-coalgebras (i.e.\ Markov chains)
$(A,\xi)$ and $(B,\zeta)$ if for each set $Y\subseteq B$ and each
state $a\in A$,
\begin{equation*}
  \zeta(f(a))(Y)=\xi(a)(f^{-1}[Y]),
\end{equation*}
i.e.\ the probability of reaching $Y$ from $f(a)$ is the same as that
of reaching $f^{-1}[Y]$ from~$a$. Morphisms of probabilistic
transition systems, viewed as coalgebras, satisfy a similar condition
for the successor distributions, and additionally preserve the truth
values of propositional atoms (formal details are given in
Remark~\ref{rem:coalg-morph-def}).

An $F$-coalgebra $(Z,\zeta)$ is \emph{final} if for every
$F$-coalgebra $(A,\xi)$ there exists exactly one morphism
$(A,\xi)\to(Z,\zeta)$. Final coalgebras are unique up to isomorphism
if they exist, and should be thought of as having as states all
possible behaviours of states in $F$-coalgebras. For our present
purposes, we do not need an explicit description of the final
coalgebra; it suffices to know that since the functor describing
probabilistic transition systems is \emph{accessible} (more precisely
$\omega_1$-accessible), a final coalgebra for it, i.e.\ a final
probabilistic transition system, exists~\cite{Barr93}.

\section{Behavioural distances and games\\ for probabilistic transition systems}
\label{sec:games}

\noindent We will now discuss several notions of behavioural distance
for probabilistic transition systems: via fixed point iteration \`a la
Wasserstein/Kantorovich, via games and via modal logics. We will
mostly focus on depth-$n$ distances, defining the unbounded distance
only for one of the cases in order to be able to formulate our main
target result, which characterizes first-order formulas that are
non-expansive w.r.t.\ unbounded-depth behavioural distance. We will
eventually show (Section~\ref{sec:modal-approx}) that at finite depth,
all these distances coincide. It has been shown in previous
work~\cite{dgjp:metrics-labelled-markov,BreugelWorrell05,BreugelEA08}
that the two unbounded-depth distances defined via Kantorovich fixed
point iteration and via the logic, respectively, coincide in very
similar settings; such results can be seen as probabilistic variants
of the Hennessy-Milner theorem.

We start by defining requisite standard notions regarding
pseudo-metric spaces.

\begin{defn}[Pseudometric spaces, non-expansive maps, function spaces]
  \label{def:metric}
  Given a non-empty set $X$, a \emph{(bounded) pseudometric on $X$} is
  a function $d\colon X\times X\to [0,1]$ such that for all
  $x,y,z\in X$, the following axioms hold: $d(x,x) = 0$
  (\emph{reflexivity}), $d(x,y) = d(y,x)$ (\emph{symmetry}),
  $d(x,z) \le d(x,y)+d(y,z)$ (\emph{triangle inequality}). If
  additionally $d(x,y)=0$ implies $x=y$, then $d$ is a
  \emph{metric}. A \emph{(pseudo)metric space} is a pair $(X,d)$ where
  $X$ is a set and $d$ is a (pseudo)metric on $X$. We equip the unit
  interval $[0,1]$ with the standard Euclidean distance~$d_e$,
  $d_e(x,y)=|x-y|$.

  A function $f\colon X\to Y$ between pseudometric spaces $(X,d_1)$,
  $(Y,d_2)$ is \emph{non-expansive} if
  $d_2 \circ (f \times f) \le d_1$, i.e.\ $d_2(f(x),f(y))\le d_1(x,y)$
  for all $x,y$. We then write $$f \colon \nonexp{(X, d_1)}{(Y, d_2)}.$$
  The space of non-expansive functions $\nonexp{(X, d_1)}{(Y, d_2)}$
  is equipped with the \emph{supremum (pseudo)metric} $d_\infty$
  defined by
  \begin{equation*}
    d_\infty(f,g) = \sup_{x\in X} d_2(f(x),g(x))
  \end{equation*}
  In the special case $(Y,d_2) = ([0,1],d_e)$, we will also denote
  $d_\infty(f,g)$ as $\supnorm{f-g}$.
  
  We denote by
  $\ball{d}{\epsilon}{a} = \{x\in X\mid d(a,x) \le \epsilon\}$ the
  \emph{ball} of radius $\epsilon$ around $a$ in $(X,d)$. The space
  $(X,d)$ is \emph{totally bounded} if for every $\epsilon > 0$ there
  exists a finite \emph{$\epsilon$-cover}, i.e.\ finitely many
  elements $a_1,\dots,a_n\in X$ such that
  $X = \bigcup_{i=1}^n \ball{d}{\epsilon}{a_i}$.
\end{defn}
\noindent Recall that a metric space is compact iff it is complete and
totally bounded.
\begin{rem}
  \label{rem:metric-termination}
  Whenever we have a bounded pseudometric $d$ on some set $X$, we can
  construct a (bounded) pseudometric $\bar d$ on the set $X+1$ by
  defining
  \begin{align*}
    \bar d(\ast,\ast) & = 0 \\
    \bar d(\ast,x_1) = \bar d(x_1,\ast) & = 1 \\
    \bar d(x_1,x_2) & = d(x_1,x_2)
  \end{align*}
  for all $x_1,x_2\in X$. We will usually omit the bar in $\bar d$ and
  just use $d$ to denote either pseudometric.
\end{rem}

\noindent Next we introduce the notions of Wasserstein and Kantorovich
lifting, which coincide according to the Kantorovich-Rubinstein
duality. To this end, we first need the notion of a coupling of two
probability distributions, from which the original distributions are
factored out as marginals.

\begin{defn}
  \label{def:coupling}
  Let $\pi_1$ and $\pi_2$ be discrete probability measures on $A$ and
  $B$, respectively. We denote by $\pi_1\otimes\pi_2$ the set of
  \emph{couplings} of~$\pi_1$ and~$\pi_2$, i.e.\ probability measures
  $\mu$ such that $\pi_1$ and $\pi_2$ are \emph{marginals} of $\mu$:
  \begin{itemize}
  \item for all $a\in A$, $\sum_{b\in B}\mu(a,b) = \pi_1(a)$;
  \item for all $b\in B$, $\sum_{a\in A}\mu(a,b) = \pi_2(b)$.
  \end{itemize}
\end{defn}

\begin{defn}[Wasserstein and Kantorovich distances]
  \label{def:wasserstein-kantorovich}
  Let $(X,d)$ be a pseudometric space. We define two pseudometrics on
  the set $\dfun X$, the \emph{Kantorovich distance}~$d^\uparrow$ and
  the \emph{Wasserstein distance}~$d^\downarrow$:
  \begin{gather*}
     d^\uparrow(\pi_1,\pi_2) = \quad
       \bigvee_{\mathclap{f\colon\nonexp{(X,d)}{([0,1],d_e)}}}
       \quad |\textstyle\int f \dif\pi_1 - \textstyle\int f \dif\pi_2| \\
     d^\downarrow(\pi_1,\pi_2) = \quad
       \bigwedge_{\mathclap{\mu\in\pi_1\otimes\pi_2}}
       \quad {\textstyle{\int}} d \dif\mu.
  \end{gather*}
\end{defn}

\noindent The notation $d^\uparrow,d^\downarrow$ is meant as a
mnemonic for the fact that these distances are obtained via suprema
respectively via infima.

\begin{rem}
  Intuitively these distances have the following meaning, when seen
  from the point of view of transportation theory
  \cite{v:optimal-transport}: Assume that the probability
  distributions $\pi_1,\pi_2$ denote a supply respectively demand
  present at the elements of $X$, which are nodes in a network. The
  distance between two nodes $x,y\in X$ is $d(x,y)$. Now, the
  Wasserstein distance $d^\downarrow$ on probability distributions
  looks for the optimal transport plan $\mu\colon X\times X\to [0,1]$
  specifying that $\mu(x,y)$ units are transported from $x$ to
  $y$. Then ${\textstyle{\int}} d \dif\mu$ is the cost of this
  transport plan, where we sum up the units $\mu(x,y)$ being
  transported, multiplied with the distance $d(x,y)$.  We are looking
  for the best transport plan with the minimal cost.

  The Kantorovich distance, on the other hand, considers the same
  scenario from the point of view of a logistics company that is hired
  to perform the transport. The logistics company sets prices at each
  location via a price function $f\colon X\to [0,1]$, buys surplus
  supply at $x$ for the price $f(x)$, and sells required demand at $y$
  for the price $f(y)$. Non-expansiveness of such price functions
  means that $f(x)-f(y) \le d(x,y)$ for all $x,y\in X$. Otherwise the
  logistics firm would not be hired, since it is cheaper to perform
  the transport directly. Now the Kantorovich distance is the maximal
  profit of the logistics company, where the supremum is taken over
  all possible (non-expansive) price functions.
\end{rem}

\noindent The above notions of lifting a distance on $X$ to a distance
on probability distributions over $X$ can be used to give fixed point
equations for behavioural distances on probabilistic transition
systems.

\begin{defn}[Fixed point iteration \`a la Wasserstein/Kantorovich]
  \label{def:fixed-point-iteration}
  Given a probabilistic transition system
  $\CA = (A,(p^\CA)_{p\in\At},\pi^\CA)$, we define the following chain
  of behavioural distances via fixed point iteration \`a la
  Wasserstein:
  \[ d^W_0(a,b) = 0 \quad \text{ and } \quad d^W_{n+1}(a,b) =
    \bigvee_{p\in\At}|p(a)-p(b)| \lor
    (d^W_n)^\downarrow(\pi_a,\pi_b) \] Similarly we obtain the
  following chain of behavioural distances via fixed point iteration
  \`a la Kantorovich:
  \[ d^K_0(a,b) = 0 \quad \text{ and } \quad d^K_{n+1}(a,b) =
  \bigvee_{p\in\At}|p(a)-p(b)| \lor (d^K_n)^\uparrow(\pi_a,\pi_b). \]
  We refer to these distances more succinctly as the \emph{depth-$n$
    Kantorovich} and \emph{Wasserstein distances}, respectively.
\end{defn}

\noindent In both cases, we start with the zero pseudometric, and in
the next iteration lift the pseudometric~$d_n$ from the previous step
via Wasserstein/Kantorovich. This lifted metric is then applied to the
probability distributions $\pi_a,\pi_b$ associated with $a,b$. (Note
that we are following the convention in
Remark~\ref{rem:metric-termination} in connection with the coalgebraic
representation in Section~\ref{sec:logics}, since $\pi_a,\pi_b$ could
be probability distributions that sum up to $0$ and are hence
represented by $*$). In addition we take the maximum with the supremum
over the distances for all predicates $p\in\At$.

\begin{expl}
  \label{expl:wasserstein}
  In order to illustrate the behavioural distance, in particular the
  Wasserstein distance, we consider the example probabilistic
  transition system depicted below, where
  $\epsilon\in[0,\frac{1}{2}]$. For simplicity we assume that there
  are no propositional atoms. States without outgoing transitions are
  terminating.
  \[
    \xymatrix{
      & & x \ar[ld]_{\frac{1}{2}} \ar[rd]^{\frac{1}{2}} & & & & y
      \ar[ld]_{\frac{1}{2}-\epsilon} \ar[rd]^{\frac{1}{2}+\epsilon} & \\
      & x_1 \ar[ld]_{\frac{1}{2}} \ar[rd]^{\frac{1}{2}}
      & & x_2 \ar@(dl,dr)_{1} & & y_1 \ar[ld]_{\frac{1}{2}-\epsilon}
      \ar[rd]^{\frac{1}{2}+\epsilon} & & y_2 \ar@(dl,dr)_{1} \\
      x_3 & & x_4 \ar@(ur,dr)^{1} & & y_3 & & y_4 \ar@(ur,dr)^{1} 
    }
  \]
  In fact, the depth-$3$ Wasserstein distance between $x$ and $y$ is
  $d^W_3(x,y) = \epsilon-\epsilon^2$, which can be explained as
  follows: The depth-$0$ distance is~$0$, while in the depth-$1$
  distance, terminating states have distance~$1$ from transient
  states. Moreover, it is not hard to see that
  $d^W_2(x_1,y_1) = \epsilon$, $d^W_2(x_1,y_2) = \frac{1}{2}$,
  $d^W_2(x_2,y_1) = \frac{1}{2}-\epsilon$, $d^W_2(x_2,y_2) = 0$, since
  these are exactly the probabilities for which they show differing
  behaviour.

  So in order to determine the depth-$3$ Wasserstein distance between
  $x$ and $y$,
  we have to solve the following transport problem, based on
  $d^W_2$. The supply (given by the probability distribution $\pi_x$)
  is written to the left of $x_1,x_2$ and the demand (given by the
  probability distribution $\pi_y$) is written to the right of
  $y_1,y_2$.
  \[
    \xymatrix{
      [\pi_x(x_1) = \frac{1}{2}] & x_1 \ar@{-}[r]^{\epsilon}
      \ar@{-}[rd]_(.2){\frac{1}{2}} & y_1 &
      [\pi_y(y_1) = \frac{1}{2}-\epsilon] \\
      [\pi_x(x_2) = \frac{1}{2}] &  x_2 \ar@{-}[r]_{0}
      \ar@{-}[ru]_(.8){\frac{1}{2}-\epsilon} & y_2 &
      [\pi_y(y_2) = \frac{1}{2}+\epsilon]
    }
  \]
  The best transport plan is to transport $\frac{1}{2}-\epsilon$ from
  $x_1$ to $y_1$ (where the distance is $\epsilon$), $\epsilon$ from
  $x_1$ to $y_2$ (distance $\frac{1}{2}$) and $\frac{1}{2}$ from $x_2$
  to $y_2$ (distance $0$). So, summing up, we have
  $d_3^W(x,y) = (\frac{1}{2}-\epsilon)\cdot \epsilon + \epsilon\cdot
  \frac{1}{2} + \frac{1}{2}\cdot 0 = \epsilon-\epsilon^2$.
\end{expl}

\noindent We now introduce a key tool for our technical development, a
novel bisimulation game inspired by the definition of the Wasserstein
distance.

\begin{defn}[Bisimulation game]
  \label{def:bisimulation-game}
  Given probabilistic transition systems $\CA,\CB$,
  $a_0\in A,b_0\in B$, and $\epsilon_0\in[0,1]$, the
  \emph{$\epsilon_0$-bisimulation game} for $a_0$ and $b_0$ is played
  by the \emph{spoiler} $S$ and the \emph{duplicator} $D$, with rules as
  follows:
  \begin{itemize}
  \item Configurations: triples $(a,b,\epsilon)$ with states $a\in A$, $b\in B$
    and the maximal allowed deviation $\epsilon\in[0,1]$.
  \item Starting configuration: $(a_0,b_0,\epsilon_0)$
  \item Moves: in each round, $D$ first picks a probability measure
    $\mu \in \pi_a\otimes\pi_b$.
    Then, $D$ distributes the deviation $\epsilon$ over all pairs
    $(a',b')$ of successors, i.e.~picks a function $\epsilon'\colon
    A\times B\to[0,1]$ such that $\int\epsilon'\dif\mu \le \epsilon$.
    Finally, $S$ picks a pair $(a',b')$ with $\mu(a',b') > 0$ and the
    new configuration is then $(a',b',\epsilon'(a',b'))$.
  \item $D$ wins if both states are terminating or $\epsilon = 1$.
  \item $S$ wins if exactly one state is terminating and $\epsilon <
    1$.
  \item Winning condition: $|p(a)-p(b)|\le\epsilon$ for all atoms $p$.
  \end{itemize}
  
  The game comes in two variants, the \emph{(unbounded) bisimulation
  game} and the \emph{$n$-round bisimulation game}, where $n\ge 0$.
  $D$ wins if the winning condition holds before every round,
  otherwise $S$ wins. More precisely, $D$ wins the unbounded game if
  she can force infinite play and the $n$-round game once $n$ rounds
  have been played (the winning condition is not checked after the
  last round, so in particular, any $0$-round game is an immediate
  win for $D$).
\end{defn}

\begin{rem}
  The above bisimulation game is designed to fit the Wasserstein
  lifting. It differs from bisimulation games in the literature
  (e.g.~\cite{DesharnaisEA08}) in a number of salient features. A
  particularly striking aspect is that~$D$'s moves are not similar to
  $S$'s moves, and moreover~$D$ in fact moves
  before~$S$. Intuitively,~$D$ is required to commit beforehand to a
  strategy that she will use to respond to~$S$'s next move. Another
  aspect worth noting is that the distance bound~$\epsilon$ changes
  as the game is being played, a complication forced by the arithmetic
  nature of probabilistic transition systems.
\end{rem}

\noindent Based on the game we obtain the following notions of
depth-$n$ and unbounded game distance.

\begin{defn}
  \label{def:game-distance}
  Given a probabilistic transition system
  $\CA = (A,(p^\CA)_{p\in\At},\pi^\CA)$, we define the following
  \emph{depth-$n$ game distance}:
  \[ d^G_n(a,b) = \bigwedge\{ \epsilon\mid D\text{ wins the $n$-round
    bisimulation game on $(a,b,\epsilon)$}\}. \]
  Similarly the (unbounded-depth) \emph{game distance} is defined by
  \[ d^G(a,b) = \bigwedge\{ \epsilon\mid D\text{ wins the 
      bisimulation game on $(a,b,\epsilon)$}\}. \]
\end{defn}

\begin{expl}
  \label{expl:game}
  In order to illustrate the game we revisit the transition system
  from Example~\ref{expl:wasserstein}.  We start with the
  configuration $(x,y,\epsilon-\epsilon^2)$ and sketch a winning
  strategy for~$D$. Recall that~$D$ must in particular find a coupling
  $\mu$ with marginals $\pi_x,\pi_y$. What she can in fact do is use
  the optimal transport plan from Example~\ref{expl:wasserstein}. She
  thus takes $\mu(x_1,y_1) = \frac{1}{2}-\epsilon$,
  $\mu(x_1,y_2) = \epsilon$, $\mu(x_2,y_2) = \frac{1}{2}$, and for all
  other pairs the value of $\mu$ is $0$. She then needs to distribute
  the allowed deviation $\epsilon-\epsilon^2$ over all pairs of
  successors, which she does using the function $\epsilon'$ with
  $\epsilon'(x_1,y_1) = \epsilon$, $\epsilon'(x_1,y_2) = \frac{1}{2}$,
  $\epsilon'(x_2,y_1) = \frac{1}{2}-\epsilon$,
  $\epsilon'(x_2,y_2) = 0$ (this is is exactly the distance $d_2^W$);
  it can then be easily checked that
  $\int\epsilon'\dif\mu \le \epsilon-\epsilon^2$ (see the computation
  in Example~\ref{expl:wasserstein})

  Now $S$ picks one of the pairs $x_i,y_j$ and the game continues in
  the configuration $(x_i,y_j,\epsilon'(x_i,y_j))$. Since
  $\epsilon'(x_i,y_j)$ is the $2$-step behavioural distance between
  $x_i$ and $y_j$, the remaining game is won by $D$.
\end{expl}

\noindent Finally we define the depth-$n$ logical distance via
quantitative probabilistic modal logic as recalled in
Section~\ref{sec:logics}, restricting to formulas of rank at most $n$.

\begin{defn}
  \label{def:logical-distance}
  Given a probabilistic transition system
  $\CA = (A,(p^\CA)_{p\in\At},\pi^\CA)$, we define \emph{depth-$n$
    logical distance}~$d^L_n$ as
    \[ d^L_n(a,b) = \bigvee_{\rk(\phi)\le n} |\phi(a)-\phi(b)|. \]
\end{defn}

\noindent The equivalence of the four bounded-depth behavioural
distances introduced above will be shown in
Theorem~\ref{thm:modal-approx}.

\bigskip

\noindent Behavioural distance forms the yardstick for our notion of
bisimulation invariance; for definiteness:
\begin{defn}
  A quantitative, i.e.\ $[0,1]$-valued, property~$P$ of states in
  probabilistic transition systems, or a formula defining such a
  property, is \emph{bisimulation-invariant} if~$P$ is non-expansive
  w.r.t.\ the game distance, i.e.\ for states $a,b$ in probabilistic
  transition systems $\CA,\CB$, respectively,
  \begin{equation*}
    |P(a)-P(b)|\le d^G(a,b).
  \end{equation*}
  Similarly,~$P$ is \emph{depth-$n$ bisimulation invariant}, or
  \emph{finite-depth bisimulation invariant} if mention of~$n$ is
  omitted, if~$P$ is non-expansive w.r.t.\ $d^G_n$ in the same
  sense. 
\end{defn}
\noindent It is easy to see that quantitative probabilistic modal
formulas are bisimulation-invariant, or as a slogan
\begin{quote}
  \emph{quantitative probabilistic modal logic is
    bisimulation-invariant,}
\end{quote} 
We refrain from giving an explicit proof, as the results of the next
section (Remark~\ref{rem:invariance}) will imply that quantitative
probabilistic modal logic is in fact even finite-depth
bisimulation-invariant (a stronger invariance property since clearly
$d^G_n\le d^G$).

\section{Modal Approximation at Finite Depth}
\label{sec:modal-approx}

We proceed to establish the most important stepping stone on the way
to the eventual van Benthem theorem: We show that every quantitative
property of states in probabilistic transition systems that is
non-expansive w.r.t.\ bounded-depth behavioural distance can be
approximated by modal formulas of bounded rank. We prove this
simultaneously with coincidence of the various finite-depth
behavioural pseudometrics defined in the previous section.

To begin, we prove separately that the pseudometric~$d^G_n$ defined by
our bisimulation game coincides with the Wasserstein
pseudometric~$d^W_n$:

\begin{lem}\label{lem:metrics-equal-gw}
  We have $d^G_n = d^W_n$.
\end{lem}
\begin{proof}
  Induction over~$n$. The base case $n = 0$ is clear: the $0$-round
  game is an immediate win for~$D$, so $d^G_0 = d^W_0 = 0$. We proceed
  with the inductive step from~$n$ to~$n+1$.
  
  So let~$a$ and~$b$ be states in probabilistic transition
  systems~$\CA,\CB$, respectively. If $a$ and $b$ are both
  terminating, then $d^G_{n+1}(a,b) = d^W_{n+1}(a,b) = 0$. If exactly
  one of $a,b$ is terminating, then
  $d^G_{n+1}(a,b) = d^W_{n+1}(a,b) = 1$. Now assume that both $a$ and
  $b$ are transient.

  ``$\ge$'': Let $d^G_{n+1}(a,b) \le \epsilon$, so $D$ wins the
  $(n+1)$-round bisimulation game on $(a,b,\epsilon)$. We show that
  $d^W_{n+1}(a,b) \le \epsilon$.
  First, for every $p\in\At$, $|p(a)-p(b)| \le \epsilon$ by the
  winning condition.
  Second, let $\mu\in\pi_a\otimes\pi_b,\epsilon'\colon A\times
  B\rightarrow[0,1]$ be $D$'s choices in the first turn.
  By assumption, $D$ wins the $n$-round bisimulation game on
  $(a',b',\epsilon'(a',b'))$ for every $a'\in A, b'\in B$, so
  $d^W_n = d^G_n\le\epsilon'$ by induction, and thus $\int
  d^W_n\dif\mu \le \int \epsilon'\dif\mu \le \epsilon$.

  ``$\le$'': Let $d^W_{n+1}(a,b) < \epsilon$. It suffices to give a
  winning strategy for $D$ in the $(n+1)$-round bisimulation game on
  $(a,b,\epsilon)$ (implying~$d^G_{n+1}(a,b)\le\epsilon$).  The
  winning condition in the initial configuration follows immediately
  from the assumption.  Also by the assumption, there exists
  $\mu\in\pi_a\otimes\pi_b$ such that $\int d^W_n\dif\mu<\epsilon$.
  As $\pi_a$ and $\pi_b$ are discrete, the set
  \[ R := \{(a',b') \mid \pi_a(a') > 0 \text{ and } \pi_b(b') > 0 \}
  \]
  is countable; so we can write $R = \{(a_1,b_1),(a_2,b_2),\dots\}$.
  Now put $\delta = \epsilon - \int d^W_n\dif\mu$ and define
  $\epsilon'(a_i,b_i) = d^W_n(a_i,b_i) + 2^{-i}\delta$ for
  $(a_i,b_i)\in R$ and $\epsilon'(a',b') = 0$ for $(a',b') \notin
  R$. Then
  \[ \textstyle \int \epsilon'\dif\mu \le \int d^W_n\dif\mu + \delta =
     \epsilon, \]
  so playing $\mu$ and $\epsilon'$ constitutes a legal move for $D$.
  Now, since $\mu\in\pi_a\otimes\pi_b$, $\mu(a',b') = 0$ for all
  $(a',b')\notin R$, $S$ must pick some $(a_i,b_i) \in R$. Then
  \[ d^G_n(a_i,b_i) = d^W_n(a_i,b_i) < \epsilon'(a_i,b_i), \]
  so $D$ wins the $n$-round game on $(a_i,b_i,\epsilon'(a_i,b_i))$.
\end{proof}

\noindent The coincidence of the remaining pseudometrics is proved
in one big induction (following a similar structure as Wild et
al.~\cite{WildEA18}), along with total boundedness (needed later in
this Section to apply a variant of the Arzel\`a-Ascoli theorem) and the
mentioned modal approximability of depth-$n$ bisimulation-invariant
properties. We phrase the latter as density of the modal formulas of
rank at most~$n$ in the non-expansive function space
(Definition~\ref{def:metric}):

\begin{thm}\label{thm:modal-approx}
  Let $\CA$ be a probabilistic transition system with state
  set~$A$. Then for all $n\ge 0$,
  \begin{enumerate}
    \item we have $d^G_n = d^W_n = d^K_n = d^L_n =: d_n$ on $\CA$;
      \label{item:metrics-equal}
    \item the pseudometric space $(A,d_n)$ is totally bounded; and
      \label{item:tot-bounded}
    \item $\modf{n}$ is a dense subset of
      $\nonexp{(A,d_n)}{([0,1],d_e)}$.
      \label{item:modal-approx}
  \end{enumerate}
\end{thm}
\begin{rem}
  As indicated in the introduction, van Breugel and
  Worrell~\cite{BreugelWorrell05} show similar but unbounded-rank
  versions of two of the claims in the above theorem, namely
  coincidence of an unbounded Kantorovich-style distance and an
  unbounded logical distance, and density of the set of \emph{all}
  modal formulas in the space of non-expansive functions w.r.t.\
  \emph{unbounded} behavioural distance, both essentially amounting to
  a quantitative Hennessy-Milner theorem. As discussed in the
  introduction, the bounds on modal rank and bisimulation depth are
  key features in the above theorem.
\end{rem}
\begin{rem}\label{rem:invariance}
  From Theorem~\ref{thm:modal-approx}, it is immediate that as claimed
  at the end of Section~\ref{sec:games}, \emph{quantitative
    probabilistic modal logic is finite-depth bisimulation-invariant}:
  By definition of $d^L_n$, every $\phi\in\modf{n}$ is
  $d^L_n$-invariant, and hence invariant w.r.t.\ all other
  finite-depth behavioural distances.
\end{rem}
\begin{proof}[Proof of Theorem~\ref{thm:modal-approx}] By
  Lemma~\ref{lem:metrics-equal-gw}, for Item~\ref{item:metrics-equal}
  we just have to show that $d^W_n = d^K_n = d^L_n$. We proceed to
  prove all claims simultaneously by induction on $n$.
  
  For the base case $n = 0$, the behavioural distances are all
  trivial: $d^W_0 = 0$ and $d^K_0 = 0$ by definition and $d^L_0 = 0$
  because all formulas of rank~$0$ are (propositional combinations of)
  constants.  In particular, the space $(A,d_0)$ is totally bounded.
  Finally, because every function that is non-expansive wrt.~$d_0$
  must be constant, Item~\ref{item:modal-approx} follows by density of
  $\RatI$ in $[0,1]$ since the modal syntax includes constants
  $c\in\RatI$.
  
  The induction step is distributed across a number of lemmas, stated
  and proved next.
\end{proof}

\noindent For the remainder of this section, we fix a model $\CA$ and
$n > 0$, and assume as the inductive hypothesis that all claims in
Theorem~\ref{thm:modal-approx} hold for all $n' < n$.

For any pseudometric space $(X,d)$, the relation
$x\sim y :\iff d(x,y) = 0$ is an equivalence relation. The quotient
set $X/{\sim}$ is made into a metric space $(X/\sim,d')$, the
\emph{metric quotient} of $(X,d)$, by taking $d'([x],[y]) = d(x,y)$.

We need the following version of the Kantorovich-Rubinstein
duality~\cite[Proposition 11.8.1]{dudley2002}:
\begin{lem}[Kantorovich-Rubinstein duality] \label{lem:kr-duality} Let
  $(X,d)$ be a separable metric space, and let $\mathcal{P}_1(X)$
  denote the space of probability measures
  $\mu\colon\mathcal{B}(X) \to [0,1]$ on the Borel $\sigma$-algebra
  $\mathcal{B}(X)$ such that
  $\textstyle{\int} d(x,\,\cdot\,) \dif\mu < \infty$ for some $x\in
  X$. Then for $\mu_1,\mu_2\in\mathcal{P}_1(X)$,
  \begin{equation*}
    \bigwedge_{\mathclap{\mu\in\mu_1\otimes\mu_2}}
    {\textstyle{\int}} d \dif\mu
    \quad = \quad
    \bigvee_{\mathclap{f\colon\nonexp{(X,d)}{([0,1],d_e)}}}
    \quad
    |\textstyle\int f \dif\mu_1 - \textstyle\int f \dif\mu_2|.
  \end{equation*}
\end{lem}
\noindent (Recall from Definition~\ref{def:coupling} that
$\mu_1\otimes\mu_2$ is the set of couplings of $\mu_1,\mu_2$.) Using
this equality, we obtain coincidence of the Kantorovich and
Wasserstein distances:
\begin{lem}
  We have $d^W_n = d^K_n$ on $\CA$.
\end{lem}
\begin{proof}
  Essentially, we need to transfer Kantorovich-Rubinstein duality to
  the slightly more general case of pseudometrics. Explicitly, let
  $(B,d)$ be the metric quotient of $(A,d_{n-1})$, and let
  $p\colon A \to B$ be the projection map. By construction, $p$ is an
  isometry.  Both the Kantorovich and the Wasserstein lifting preserve
  isometries~\cite{bbkk:behavioral-metrics-functor-lifting}, so for
  all discrete probability measures $\pi_1,\pi_2$ on $A$,
  \begin{align*}
    (d_{n-1})^{\uparrow}(\pi_1,\pi_2) 
    & = d^{\uparrow}((\dfun p)\pi_1,(\dfun p)\pi_2) \\
    & = d^{\downarrow}((\dfun p)\pi_1,(\dfun p)\pi_2) \\
    & = (d_{n-1})^{\downarrow}(\pi_1,\pi_2).
  \end{align*}
  In the second step we have applied Lemma~\ref{lem:kr-duality} to the
  metric space $(B,d)$, noting that every discrete probability measure
  can be defined on the Borel $\sigma$-algebra and every totally
  bounded space is separable.
\end{proof}
\noindent For the purpose of relating the logical distance to the
Kantorovich distance, we note next that our modality~$\Diamond$ is
non-expansive. Explicitly, we extend~$\Diamond$ to act on
non-expansive functions $f\colon\nonexp{(A,d_n)}{([0,1],d_e)}$ by
\begin{align*}
  \Diamond f\colon\nonexp{(A,d_n)&}{([0,1],d_e)}\\
  (\Diamond f)(a)&=\intsuc{a}{f}.
\end{align*}
\begin{lem}\label{lem:diamond-nonexp}
  The map $f \mapsto \Diamond f$ is non-expansive w.r.t.\ the supremum
  metric.
\end{lem}
\begin{proof}
  Let $\supnorm{f-g}\le\epsilon$; we have to show
  $\supnorm{\Diamond f - \Diamond g}\le\epsilon$. So let $a\in
  A$; then
  \[
    |(\Diamond f)(a) - (\Diamond g)(a)|
    = \intsuc{a}{(f-g)}
    \le \intsuc{a}{\epsilon} \le \epsilon,
  \]
  as required.
\end{proof}
\noindent This allows us to discharge the remaining equality of
behavioural distances, regarding the Kantorovich distance and the
logical distance:
\begin{lem}
  We have $d^K_n = d^L_n$ on $\CA$.
\end{lem}
\begin{proof}
  Let $a,b\in A$ and consider the map
  \[
    G \colon \nonexp{(F,d_\infty)}{([0,1],d_e)}, \quad
    f \mapsto |(\Diamond f)(a) - (\Diamond f)(b)|,
  \]
  where $F = \nonexp{(A,d_{n-1})}{([0,1],d_e)}$. Then~$G$ is a
  continuous function because all of its constituents are continuous
  (in particular, $\Diamond$ is continuous by
  Lemma~\ref{lem:diamond-nonexp}).
  
  By the induction hypothesis, and because density is preserved by
  continuous maps, $G[\modf{n-1}]$ is a dense subset of $G[F]$.
  Thus,
  \begin{align*}
    d^K_n(a,b)
    & = \bigvee_{p\in\At} |p(a)-p(b)| \lor \bigvee G[F]
    = \bigvee_{p\in\At} |p(a)-p(b)| \lor \bigvee G[\modf{n-1}] \\
    & = \bigvee_{p\in\At} |p(a)-p(b)| \lor
    \bigvee_{\rk\phi\le n-1} |(\Diamond\phi)(a)-(\Diamond\phi)(b)|
    = \bigvee_{\rk\phi\le n} |\phi(a)-\phi(b)| = d^L_n(a,b).
  \end{align*}
  To prove the penultimate step, we first note that ``$\le$'' follows
  immediately. To see ``$\ge$'', we proceed by induction over the
  Boolean combinations of atoms $p\in\At$ and formulas $\Diamond\phi$,
  where $\phi\in\modf{n-1}$, using that for any formulas $\phi,\psi$
  and $c\in\RatI$:

  \begin{align*}
    |(\phi\ominus c)(a)-(\phi\ominus c)(b)| & \le |\phi(a)-\phi(b)| \\
    |(\neg\phi)(a)-(\neg\phi)(b)| & = |\phi(a)-\phi(b)| \\
    |(\phi\land\psi)(a)-(\phi\land\psi)(b)| &
    \le |\phi(a)-\phi(b)| \lor |\psi(a)-\psi(b)|.\qedhere
  \end{align*}
\end{proof}

\noindent We make use of the following two lemmas~\cite{WildEA18},
which are versions of the Arzel\`a-Ascoli theorem and the
Stone-Weierstra\ss{} theorem where function spaces are restricted to
non-expansive functions instead of the more general continuous
functions, but the underlying spaces are only required to be totally
bounded instead of compact:

\begin{lem}[Arzel\`a-Ascoli for totally bounded
  spaces]\label{lem:arzela-ascoli}
  Let $(X,d_1),(Y,d_2)$ be totally bounded pseudometric spaces. Then
  the space $\nonexp{(X,d_1)}{(Y,d_2)}$, equipped with the supremum
  pseudometric, is totally bounded.
\end{lem}

\begin{lem}[Stone-Weierstra\ss{} for totally bounded spaces]
  \label{lem:stone-weierstrass}
  Let $(X,d)$ be a totally bounded pseudometric space, and let $L$ be
  a subset of $F := \nonexp{(X,d)}{([0,1],d_e)}$ such that
  $f_1,f_2 \in L$ implies $\min(f_1,f_2),\max(f_1,f_2) \in L$.
  Then $L$ is dense in $F$ if each $f\in F$ can be approximated at each
  pair of points by functions in~$L$; that is for all $\epsilon>0$ and
  all $x_1,x_2\in X$ there exists $g\in L$ such that
  $\max(|f(x_1)-g(x_1)|,|f(x_2)-g(x_2)|) \le\epsilon$.
\end{lem}

\noindent We use the Arzel\`a-Ascoli theorem to complete the inductive
step for Item~\ref{item:tot-bounded} in
Theorem~\ref{thm:modal-approx}. This is done in the following lemma --
related lemmas have already appeared in~\cite{WildEA18} (for a fuzzy
powerset functor) and~\cite{km:bisim-games-logics-metric} (for general
functors).

\begin{lem}\label{lem:tot-bounded}
  $(A,d_n)$ is a totally bounded pseudometric space.
\end{lem}
\begin{proof}[Proof (sketch)]
  By the induction hypothesis, and using
  Lemma~\ref{lem:arzela-ascoli}, we know that $F :=
  \nonexp{(A,d_{n-1})}{([0,1],d_e)}$ is totally bounded.
  
  Let $\epsilon > 0$. As $\modf{n-1}$ is dense in $F$, there exist
  finitely many $\phi_1,\dots,\phi_m\in\modf{n-1}$ such that
  $\bigcup_{i=1}^m \ball{}{\frac{\epsilon}{8}}{\phi_i} = F$. From these
  formulas, together with the propositional atoms $p_1,\dots,p_k$, we
  can construct the map
  \begin{equation*}
    I\colon A \to [0,1]^{k+m}, \quad
    a \mapsto
    (p_1(a),\dots,p_k(a),(\Diamond\phi_1)(a),\dots,(\Diamond\phi_m)(a)).
  \end{equation*}
  It turns out that $I$ is an $\frac{\epsilon}{4}$-isometry, i.e.\
  $|d_n(a,b) - \supnorm{I(a) - I(b)}| \le \tfrac{\epsilon}{4}$ for all
  $a,b\in A$. Thus, by the triangle inequality, we can take preimages
  to turn a finite $\frac{\epsilon}{4}$-cover of $[0,1]^{k+m}$ (a
  compact, hence totally bounded space) into a finite $\epsilon$-cover
  of $(A,d_n)$.
\end{proof}
\noindent This covers the total boundedness claim in
Theorem~\ref{thm:modal-approx}, and subsequently enables us to use the
above version of the Stone-Weierstra\ss{} theorem
(Lemma~\ref{lem:stone-weierstrass}) to prove the density claim:

\begin{lem}
  $\modf{n}$ is a dense subset of $\nonexp{(A,d_n)}{([0,1],d_e)}$.
\end{lem}
\begin{proof}[Proof (sketch)]
  By Lemma~\ref{lem:tot-bounded}, $(A,d_n)$ is totally bounded; since
  moreover $\modf{n}$ is closed under $\land$ and $\lor$, we can apply
  Lemma~\ref{lem:stone-weierstrass}.
  
  It thus remains to give, for each non-expansive map
  $f\colon\nonexp{(A,d_n)}{([0,1],d_e)}$, states $a,b \in A$ and
  $\epsilon > 0$, a formula $\phi\in\modf{n}$ such that
  $|f(a)-\phi(a)| \le \epsilon$ and $|f(b)-\phi(b)| \le \epsilon$.
  
  To construct such a formula, we note that $|f(a)-f(b)|\le
  d^L_n(a,b)$ (by non-expansiveness), so there exists some
  $\psi\in\modf{n}$ such that $|\psi(a)-\psi(b)|\ge
  |f(a)-f(b)|-\epsilon$. The desired formula $\phi$ can now be
  constructed from $\psi$ with the help of truncated subtraction $\ominus$.
\end{proof}

\noindent This completes the proof of
Theorem~\ref{thm:modal-approx}. Now that we have a way to approximate
depth-$k$ bisimulation-invariant properties by modal formulas of rank
$k$, on any fixed model, we need a way to make such an approximation
uniform across all possible models. Put differently, we need a
probabilistic transition system that realizes all behaviours up to
depth-$n$ bisimilarity. Unlike in the fuzzy setting~\cite{WildEA18}
(where a final system fails to exist for cardinality reasons), we can
use a final coalgebra, i.e.\ a final probabilistic transition system,
for this purpose. We have recalled in Section~\ref{sec:coalg} that
such a final probabilistic transition system, denoted~$\CF$ in the
following, exists. As explained in Section~\ref{sec:coalg},\;\;$\CF$
even realizes all behaviours up to bisimilarity, in a sense that we
will make more precise presently. Recall that~$\CF$ is formally
characterized by admitting a unique morphism from every probabilistic
transition system. We recall the preservation properties of such
morphisms from Section~\ref{sec:coalg}:

\begin{rem}\label{rem:coalg-morph-def}
  If $f\colon\CA\to\CB$ is a morphism of probabilistic transition
  systems (seen as coalgebras), then, by unfolding definitions, we see
  that for all $a\in A$:
  \begin{itemize}
    \item for all $p\in\At$, $p^\CA(a) = p^\CB(f(a))$,
    \item $a$ is terminating $\iff$ $f(a)$ is terminating,
    \item for all $b'\in B$:
      $\pi^\CB(f(a),b') = \sum_{f(a')=b'} \pi^\CA(a,a')$.
  \end{itemize}
\end{rem}
\noindent Using these properties, we see that morphisms preserve
behaviour on-the-nose, that is:
\begin{lem}\label{lem:bounded-morphism-game}
  Let $f\colon\CA\to\CB$ be a coalgebra morphism. Then, for any $a_0\in
  A$, $d^G(a_0,f(a_0)) = 0$.
\end{lem}
\begin{proof}
  We show that $D$ wins the bisimulation game for $(a_0,f(a_0),0)$ by
  maintaining the invariant that the current configuration is of the
  form $(a,b,0)$ with $b = f(a)$. By Remark~\ref{rem:coalg-morph-def},
  this ensures that the winning condition always holds. It remains to
  show that $D$ can maintain the invariant.

  In each round, $D$ begins by picking $\mu(a',b') = \pi_a(a')$ if
  $b'=f(a')$ and $0$ otherwise, and $\epsilon' = 0$.
  We can see that $\mu\in\pi_a\otimes\pi_b$, because, still following
  Remark~\ref{rem:coalg-morph-def},
  \begin{equation*}
    \sum_{b'\in B}\mu(a',b') = \pi_a(a') \quad \text{and} \quad
    \sum_{a'\in A}\mu(a',b') = \sum_{f(a')=b'} \pi_a(a') = \pi_b(b')
  \end{equation*}
  for all $a'\in A$ and $b'\in B$. Also, clearly $\int\epsilon'\dif\mu
  = 0$. Now any choice of $S$ leads to another configuration
  $(a',b',0)$ with $b'=f(a')$.
\end{proof}
\noindent This entails the following lemma, which will enable us to
use approximants on the final probabilistic transition system~$\CF$ as
uniform approximants across all models:
\begin{lem}\label{lem:uniform-approx}
  Let $\phi$ and $\psi$ be bisimulation-invariant first-order
  properties. Then, for any model $\CA$,
  $\supnorm{\phi-\psi}^\CA \le \supnorm{\phi-\psi}^\CF$.
\end{lem}
\begin{proof}
  Let $\CA$ be a model, and let $h\colon\CA\to\CF$ be the unique
  morphism. Let $a\in A$. Then $d^G(a,h(a)) = 0$ by
  Lemma~\ref{lem:bounded-morphism-game}, and thus
  $\phi_\CA(a) = \phi_\CF(h(a))$ and $\psi_\CA(a) = \psi_\CF(h(a))$ by
  bisimulation invariance. So
  \begin{equation*}
    \supnorm{\phi-\psi}^\CA =
    \sup_{a\in A} |\phi_\CA(a)-\psi_\CA(a)| =
    \sup_{a\in A} |\phi_\CF(h(a))-\psi_\CF(h(a))| \le
    \supnorm{\phi-\psi}^\CF. \qedhere
  \end{equation*}
\end{proof}

\section{Locality}\label{sec:locality}

As indicated in the introduction, the proof of our van Benthem theorem
now proceeds by first establishing that every bisimulation-invariant
first-order formula~$\phi$ is \emph{local} in a sense to be made
precise shortly, and subsequently that~$\phi$ is in fact even
finite-depth bisimulation invariant, for a depth that is exponential
in the rank of~$\phi$. The announced notion of locality makes
reference to a notion of Gaifman graph~\cite{Gaifman82} and distance
that we adapt to the probabilistic setting:

\begin{defn}
  Let $\CA$ be a probabilistic transition system.
  \begin{enumerate}
  \item The \emph{Gaifman graph} of $\CA$ is the undirected graph on
    the set~$A$ of vertices that has an edge for every pair $(a,a')$
    with $\pi(a,a') > 0$.
    \item The \emph{Gaifman distance} $D \colon A\times A \to
      \mathbb{N}\cup\{\infty\}$ is graph distance in the Gaifman
      graph: for every $a,a'\in A$, the distance $D(a,a')$ is the
      least number of edges on a path from $a$ to $a'$, if at least
      one such path exists, and $\infty$ otherwise.
    \item For $a\in A$ and $k \ge 0$, the \emph{radius $k$
      neighbourhood} of $a$ in $\CA$, denoted by $\nbhood{k}{a}$, is
      the subset of $A$ that is reachable in at most $k$ steps:
      $\nbhood{k}{a} = \{a' \in A \mid D(a,a') \le k \}$.
      For $\bar a = (a_1,\dots,a_n)$ we put
      $\nbhood{k}{\bar a} = \bigcup_{i\le n}\nbhood{k}{a_i}$.
  \end{enumerate}
\end{defn}

\noindent Given a state $a$ in a probabilistic transition system~$\CA$
and a radius $k$, we can now restrict~$\CA$ to a smaller set of states
by discarding all states at a distance greater than $k$
from~$a$. States at distance $k$ become terminating. Formally:

\begin{defn}
  Let $\CA$ be a model, $a\in A$ and $k\ge0$. The \emph{restriction}
  of $\CA$ to $\nbhood{k}{a}$ is the model $\CA^k_a$ with set
  $\nbhood{k}{a}$ of states, and
  \begin{align*}
    p^{\CA^k_a}(b) & = p^\CA(b) \\
    \pi^{\CA^k_a}(b,c) & =
    \begin{cases}
      \pi^\CA(b,c), & \text{ if } D(a,b) < k, \\
      0, & \text{ if } D(a,b) = k, \\
    \end{cases}
  \end{align*}
  for all $p\in\At$ and $b,c\in\nbhood{k}{a}$. Note that this does
  actually define a probabilistic transition system, because if
  $D(a,b) < k$, then $D(a,c)\le k$ for all $c$ with
  $\pi^\CA(b,c) > 0$.
\end{defn}

\noindent These restricted models have the expected relationship with
games of bounded depth:

\begin{lem}\label{lem:nbhood-bisim}
  Let~$a$ be a state in a probabilistic transition
  system~$\CA$. Then~$D$ wins the $k$-round $0$-bisimulation game for
  $\CA,a$ and $\CA^k_a,a$.
\end{lem}

\begin{proof}
  Player~$D$ wins by maintaining the invariant that whenever $i$
  rounds have been played, the current configuration is of the form
  $(a_i,a_i,0)$ for some $a_i\in A$ with $D(a,a_i)\le i$.  For $i<k$,
  no configuration of this kind can be winning for $S$, because the
  two states in this configuration represent the same state in
  different models (recall that the winning conditions are not checked
  after the last round has been played).

  It remains to give a strategy for $D$ that maintains the invariant.
  It clearly holds at the start of the game, with $a_0 = a$.
  When the $(i+1)$-th round is played, $D$ can pick
  $\mu\in\pi_{a_i}\otimes\pi_{a_i}$ and $\epsilon'\colon
  A\times\nbhood{k}{a} \to [0,1]$ as follows:
  \begin{align*}
    \mu(a',a'') & =
    \begin{cases}
      \pi_{a_i}(a'), & \text{ if } a' = a'', \\
      0, & \text{ otherwise},
    \end{cases} \\
    \epsilon'(a',a'') & = 0.
  \end{align*}
  Clearly, $\int \epsilon'\dif\mu = 0$, so this is a legal move. Now
  the new configuration chosen by $S$ necessarily satisfies the
  invariant.
\end{proof}

\noindent Locality of a formula now means that its truth values only
ever depend on the neighbourhood of the state in question:

\begin{defn}
  A formula $\phi(x)$ is \emph{$k$-local} for some radius $k$, if for
  every model $\CA$ and every state $a\in A$,
  $\phi_\CA(a) = \phi_{\CA^k_a}(a)$.
\end{defn}
\noindent Since modal formulas are bisimulation-invariant,
Lemma~\ref{lem:nbhood-bisim} implies

\begin{lem}
  Every quantitative probabilistic modal formula of rank at most $k$
  is $k$-local.
\end{lem}
\noindent To prove locality of bisimulation-invariant first-order
formulas, we require a model-theoretic tool, an adaptation of
Ehrenfeucht-Fra\"iss\'e equivalence to the probabilistic setting:

\begin{defn}
  \label{def:ef-game}
  Let $\CA,\CB$ be probabilistic transition systems, and
  let~$\bar a_0$ and~$\bar b_0$ be vectors of equal length over $A$
  and $B$, respectively. The \emph{Ehrenfeucht-Fra\"iss\'e game for
    $\CA,\bar a_0$ and $\CB,\bar b_0$}, played by $S$ (\emph{spoiler})
  and $D$ (\emph{duplicator}), is given as follows.
  \begin{itemize}
    \item{Configurations:} pairs $(\bar a,\bar b)$ of vectors $\bar a$
      over $A$ and $\bar b$ over $B$.
    \item \emph{Initial configuration:} $(\bar a_0,\bar b_0)$.
    \item \emph{Moves:}
      Each round can be played in one of two ways, chosen by $S$:
      \begin{itemize}
      \item Standard round: $S$ may select a state in one model, say
        $a\in A$, and $D$ then has to select a state in the other
        model, say $b\in B$, reaching the configuration
        $(\bar aa,\bar bb)$.
      \item Probabilistic round: $S$ may select an index $i$ and a
        fuzzy subset in one of the models, say
        $\phi_A\colon A\to [0,1]$. $D$ then needs to select a fuzzy
        subset in the other model, say $\phi_B\colon B\to [0,1]$, such
        that $\intsuc{a_i}{\phi_A} = \intsuc{b_i}{\phi_B}$.  Then, $S$
        selects an element on one of the sides, say $a\in A$, such
        that $\pi(a_i,a)>0$, and $D$ then selects an element on the
        other side, say $b\in B$, such that $\phi_A(a) = \phi_B(b)$
        and $\pi(b_i,b)>0$, reaching the configuration
        $(\bar aa,\bar bb)$.
      \end{itemize} 
    \item Winning conditions: Any player who cannot move loses. $S$ wins
      if a configuration is reached (including the initial configuration)
      that fails to be a partial isomorphism. Here, a
      configuration $(\bar a,\bar b)$ is a \emph{partial
      isomorphism} if
      \begin{itemize}
        \item $a_i=a_j\iff b_i=b_j$
        \item $p(a_i) = p(b_i)$ for all $i$ and all $p\in\At$
        \item $\pi_{a_i}(a_j) = \pi_{b_i}(b_j)$ for all $i,j$.
      \end{itemize}
      $D$ wins if she reaches the $n$-th round (maintaining configurations
      that are not winning for $S$).
  \end{itemize}
\end{defn}
\noindent As indicated in the related work section, our probabilistic
version of the game is partly modelled on games for topological
first-order logic~\cite{MakowskyZiegler80}, the main difference being
that in probabilistic rounds, we let the players select fuzzy instead
of crisp subsets. (Similarly, in Desharnais et al.'s probabilistic
bisimulation games~\cite{DesharnaisEA08}, the probabilistic rounds
involve crisp subsets.) For our purposes, we need only soundness of
Ehrenfeucht-Fra\"iss\'e equivalence:
\begin{lem}[Ehrenfeucht-Fra\"iss\'e invariance]
  \label{lem:ef-inv-fol}
  Let $\CA,\CB$ be probabilistic transition systems, and let
  $\bar a_0,\bar b_0$ be vectors of length $m$ over~$A$ and $B$,
  respectively.  Suppose that $D$ wins the $n$-round
  Ehrenfeucht-Fra\"iss\'e game on $\bar a_0,\bar b_0$.  Then, for
  every probabilistic first-order formula $\phi$ with at most $m$ free variables
  $x_1,\dots,x_m$ and $\qr(\phi)\le n$,
  \begin{equation*}
    \phi(\bar a_0) = \phi(\bar b_0).
  \end{equation*}  
\end{lem}
\begin{proof}
  We proceed by induction over formulas.
  \begin{itemize}
    \item The cases $p(x_i)$ and $x_i=x_j$ (with $p\in\At$) follow
      immediately from the fact that the initial configuration is a
      partial isomorphism.
    \item The Boolean cases ($c, \phi\ominus c, \neg\phi,
      \phi\land\psi$) follow directly by the induction hypothesis.
    \item $\exists x.\,\phi$:
      Let $(\bar a,\bar b)$ be the current configuration.
      Let $\delta>0$, let $a$ be such that
      \begin{equation*}
        (\exists x.\,\phi)(\bar a) - \phi(\bar a a)< \delta,
      \end{equation*}
      and let $b$ be the winning answer for $D$ in reply to $S$
      choosing $a$. By induction, $\phi(\bar a a) = \phi(\bar b b)$,
      so
      \begin{equation*}
        (\exists x.\,\phi)(\bar b)\ge\phi(\bar b b) = \phi(\bar a a)
        >(\exists x.\,\phi)(\bar a)-\delta.
      \end{equation*}
      Because $\delta>0$ was arbitrary, it follows that 
      $(\exists x.\,\phi)(\bar b)\ge(\exists x.\,\phi)(\bar a)$.
      We can symmetrically show that
      $(\exists x.\,\phi)(\bar a)\ge(\exists x.\,\phi)(\bar b)$, which
      proves this case.
    \item $\diabind{x_i}{y_{m+1}}{\phi}$: Let $(\bar a,\bar b)$ be the
      current configuration.  Suppose that $S$ picks the index $i$ and
      the fuzzy subset
      \begin{equation*}
        \phi_A\colon A\to [0,1], \quad a \mapsto \phi_\CA(\bar a a)
      \end{equation*}
      and $D$'s winning reply is $\psi_B\colon B\to [0,1]$. We show
      that on the support of $\pi_{b_i}$, $\psi_B$ must be equal to
      \begin{equation*}
        \phi_B\colon B\to [0,1], \quad b \mapsto \phi_\CB(\bar b b).
      \end{equation*}
      Suppose there exists some $b\in B$ with $\pi(b_i,b)>0$ and
      $\phi_B(b) \neq \psi_B(b)$. Then $D$ has a winning reply
      $a\in A$ in case $S$ picks this $b$, which means, by the rules
      of the game, that $\pi(a_i,a)>0$ and $\phi_A(a) =
      \psi_B(b)$.
      However, it is also true that $\phi_A(a) = \phi_B(b)$, by the
      induction hypothesis.  This is a contradiction.

      Now, because $\psi_B$ was a winning reply, we obtain
      \begin{equation*}
        (\diabind{x_i}{x_{m+1}}{\phi})(\bar a) =
        \intsuc{a_i}{\phi_A} =
        \intsuc{b_i}{\psi_B} =
        \intsuc{b_i}{\phi_B} =
        (\diabind{x_i}{x_{m+1}}{\phi})(\bar b). \qedhere
      \end{equation*}
  \end{itemize}
\end{proof}
\noindent Since embeddings into disjoint unions of models are
morphisms, the following is immediate from
Lemma~\ref{lem:bounded-morphism-game}:
\begin{lem}\label{lem:bisim-inv-disjoint}
  Every bisimulation-invariant formula is also invariant under
  disjoint union.
\end{lem}
\noindent We have now assembled the necessary ingredients to prove our
desired locality result:
\begin{lem}[Locality]\label{lem:bisim-inv-local}
  Let $\phi$ be a bisimulation invariant first-order formula of quantifier rank
  $n$ with one free variable. Then $\phi$ is $k$-local for $k = 3^n$.
\end{lem}
\begin{proof}
  Let $a$ be a state in a probabilistic transition system~$\CA$. We
  need to show $\phi_\CA(a) = \phi_{\CA^k_a}(a)$.  Let $\CB$ be a new
  model that extends $\CA$ by adding $n$ disjoint copies of both $\CA$
  and $\CA^k_a$. Let $\CC$ be the model that extends $\CA^k_a$
  likewise.  We finish the proof by showing that
  \begin{equation*}
    \phi_\CA(a) = \phi_\CB(a) = \phi_\CC(a) = \phi_{\CA^k_a}(a).
  \end{equation*}
  The first and third equality follow by bisimulation invariance of
  $\phi$ (Lemma~\ref{lem:bisim-inv-disjoint}).
  The second equality follows by Ehrenfeucht-Fra\"iss\'e invariance
  (Lemma~\ref{lem:ef-inv-fol}) once we show that $D$ has a winning
  strategy in the $n$-round Ehrenfeucht-Fra\"iss\'e game for $\CB,a$
  and $\CC,a$.
  
  Such a winning strategy can be described as follows: $D$ maintains
  the invariant that, if the configuration reached after $i$ rounds is
  $(\bar b,\bar c)$, then there exists an isomorphism $f_i$ between
  $\nbhood{k_i}{\bar b}$ and $\nbhood{k_i}{\bar c}$ that maps each
  $b_j$ to the corresponding $c_j$, where $k_i = 3^{n-i}$.
  
  The invariant holds at the start of the game, because the
  neighbourhoods on both sides are just $\nbhood{k}{a}$.
  Similarly, whenever the invariant holds, the current configuration
  is a partial isomorphism by restriction of the given isomorphism to
  the two vectors of the configuration.
  
  Now we consider what happens during the rounds. Suppose that $i$
  rounds have been played, and the current configuration is $(\bar
  b,\bar c)$. If $S$ decides to play a standard round, playing some
  $b\in B$, then there are two cases:
  \begin{itemize}
    \item $b\in \nbhood{2k_{i+1}}{\bar b}$:
      In this case, the radius-$k_{i+1}$ neighbourhood
      $\nbhood{k_{i+1}}{b}$ of $b$ is fully contained in the domain
      $\nbhood{k_i}{\bar b}$ of $f_i$ -- this follows by the triangle
      inequality, as $2k_{i+1} + k_{i+1} = 3k_{i+1} = k_i$. Now $D$
      can just reply with $c := f_i(b)$, and an isomorphism $f_{i+1}$
      between $\nbhood{k_{i+1}}{\bar bb}$ and $\nbhood{k_{i+1}}{\bar
      cc}$ is formed by restricting the domain and codomain of $f_i$
      appropriately.
    \item $b\notin \nbhood{2k_{i+1}}{\bar b}$:
      In this case, the radius-$k_{i+1}$ neighbourhoods
      $\nbhood{k_{i+1}}{b}$ of $b$ and $\nbhood{k_{i+1}}{\bar b}$ of
      $\bar b$ do not intersect -- this too follows from the triangle
      inequality.
      Now $D$ can pick a fresh copy of $\CA$ or $\CA^k_a$ in $\CC$
      (depending on which kind of copy $b$ lies in); her reply $c$ is
      then just $b$ in that copy. Here, a fresh copy is one that was
      never visited on any of the previous rounds. By construction of
      $\CB$ and $\CC$, such a copy is always available.
      This means that we now have two isomorphisms, one between
      $\nbhood{k_{i+1}}{\bar b}$ and $\nbhood{k_{i+1}}{\bar c}$ (by
      restriction of $f_i$), and one between
      $\nbhood{k_{i+1}}{b}$ and $\nbhood{k_{i+1}}{c}$ (by isomorphism
      of the respective copies of $\CA$ or $\CA^k_a$). Because these
      isomorphisms have disjoint domains and codomains, we can combine
      them to form the desired isomorphism $f_{i+1}$.
  \end{itemize}
  If $S$ plays a standard round with some $c\in C$ instead, the
  same argument applies.

  Finally, if $S$ starts a probabilistic round by picking an index
  $0\le j\le i$ and playing some $\phi_B\colon B\to [0,1]$, then we
  first note that, by the rules of the game, the support of $\phi_B$
  must be contained in $\nbhood{1}{\bar b}$, which in turn must be
  contained in the domain of $f_i$. This means that $D$ can construct
  $\phi_C\colon C\to [0,1]$ by mapping along $f_i$, i.e.~$\phi_C(c) =
  \phi_B(f_i^{-1}(c))$ for all successors $c$ of $c_j$, and
  $\phi_C(c) = 0$ otherwise. Now, whichever $b$ or $c$ is picked by
  $S$, $D$ can just reply with $c:=f_i(b)$ or $b:=f_i^{-1}(c)$ and
  $f_{i+1}$ is formed as in the first case of a standard round.
  Again, the same argument applies if $S$ picks a fuzzy subset
  $\phi_C$ on the other side.
\end{proof}

\section{A Probabilistic van Benthem Theorem}
\label{sec:main}
\noindent Having established locality of bisimulation-invariant
first-order formulas and modal approximability of finite-depth
bisimulation-invariant properties, we now discharge the last remaining
steps in our programme: We show by means of an unravelling
construction that bisimulation-invariant first-order formulas are
already finite-depth bisimulation-invariant, and then conclude our
main result, the probabilistic van Benthem theorem. 
\begin{defn}
  Let $\CA$ be a probabilistic transition system. The
  \emph{unravelling} $\CA^\ast$ of $\CA$ is a probabilistic transition
  system with non-empty finite sequences $\bar a\in A^+$ as states,
  where atoms and transition probabilities are defined as follows:
  \begin{gather*}
    p^{\CA^\ast}(\bar a) = p^\CA(\last(\bar a)) \\
    \pi^{\CA^\ast}(\bar a,\bar aa) = \pi^\CA(\last(\bar a),a),
  \end{gather*}
  for any $\bar a \in A^+$ and $a\in A$, where $\last(\bar a)$ is the
  last element of $\bar a$.
\end{defn}

\noindent As usual, models are bisimilar to their unravellings:

\begin{lem}\label{lem:bisim-unravel}
  For any probabilistic transition system $\CA$ and $a\in A$, $D$ has
  a winning strategy in the $0$-bisimulation game for $\CA,a$ and $\CA^\ast,a$.
\end{lem}
\begin{proof}
  $D$ wins by maintaining the invariant that the configuration of the
  game is of the form $(\bar a,\last(\bar a),0)$ for some $\bar a\in A^+$.
  To do so, she can put $\mu(\bar aa,a) = \pi_{\bar a}(\bar aa) =
  \pi_{\last(\bar a)}(a)$ for all $a\in A^+$, all other values of
  $\mu$ are $0$, and $\epsilon' = 0$. Then any move by $S$ leads to a
  configuration where the invariant holds.
\end{proof}
\noindent We next show that locality and bisimulation invariance imply
finite-depth bisimulation invariance:
\begin{lem}\label{lem:local-k-bisim-inv}
  Let $\phi$ be bisimulation invariant and $k$-local.
  Then $\phi$ is depth-$k$ bisimulation invariant.
\end{lem}
\begin{proof}
  Let $\CA$ and $\CB$ be two probabilistic transition systems and
  let $a\in A$ and $b\in B$ be two states such that
  $d_k^G(a,b)<\epsilon$. It is enough to show that
  $|\phi_\CA(a)-\phi_\CB(b)|\le\epsilon$.
  
  We denote by $a'$ and $a''$ the copies of $a$ in $\CA^\ast$ and
  $(\CA^\ast)^k_a$, respectively. Similarly, $b'$ and $b''$ denote the
  copies of $b$ in $\CB^\ast$ and $(\CB^\ast)^k_b$.
  By Lemma~\ref{lem:bisim-unravel}, $D$ wins the
  $0$-bisimulation-game for $\CA,a$ and $\CA^\ast,a'$ (similarly for
  $\CB$) and by Lemma~\ref{lem:nbhood-bisim}, she also wins the
  $k$-round $0$-bisimulation game for $\CA^\ast,a'$ and
  $(\CA^\ast)^k_a,a''$ (similarly for $\CB$).
  Because behavioural distance $d^G_k$ is a pseudometric, this
  means that
  \[
    d^G_k(a'',b'')\le
    d^G_k(a'',a')+d^G_k(a',a)+d^G_k(a,b)+d^G_k(b,b')+d^G_k(b',b'')
    =d^G_k(a,b)<\epsilon,
  \]
  so $D$ has a winning strategy in the $k$-round
  $\epsilon$-bisimulation game for $(\CA^\ast)^k_a,a''$ and
  $(\CB^\ast)^k_b,b''$.
  
  In both $(\CA^\ast)^k_a,a''$ and $(\CB^\ast)^k_b,b''$, the reachable
  states form a tree of depth at most $k$. This implies that, after
  $i$ rounds of the game, the two states on either side of the current
  configuration are nodes at distance $i$ from the root of their
  respective tree. Thus, whenever $k$ rounds have been played in the
  game, $S$ does not have a legal move in the next round, because at
  that point, both nodes in the configuration are necessarily leaves
  and thus terminating.
  This in turn means that if $D$ can win the $k$-round game, she
  also wins the unbounded game, so, by bisimulation invariance of
  $\phi$, $|\phi_{(\CA^\ast)^k_a}(a'')-\phi_{(\CB^\ast)^k_b}(b'')|\le\epsilon$.

  By locality and bisimulation invariance of $\phi$, and again
  Lemma~\ref{lem:bisim-unravel}, we have
  $\phi_{(\CA^\ast)^k_a}(a'') = \phi_{\CA^\ast}(a') =
  \phi_\CA(a)$
  as well as
  $\phi_{(\CB^\ast)^k_b}(b'') = \phi_{\CB^\ast}(b') =
  \phi_\CB(b)$.
  Thus $|\phi_\CA(a)-\phi_\CB(b)|\le\epsilon$, as claimed.
\end{proof}
\noindent Our main result is then stated as follows:
\begin{thm}[Probabilistic van Benthem theorem]
  Every bisimulation-invariant formula of probabilistic first order
  logic with rank at most $n$ can be approximated (uniformly across
  all models) by probabilistic modal formulas of rank at most $3^n$.
\end{thm}
\begin{proof}
  Let $\phi$ be a probabilistic first-order formula of rank $n$. By
  Lemma~\ref{lem:bisim-inv-local} and
  Lemma~\ref{lem:local-k-bisim-inv}, $\phi$ is
  depth-$k$ bisimulation-invariant for $k = 3^n$.
  By Theorem~\ref{thm:modal-approx}, for every $\epsilon>0$, there
  exists a probabilistic modal formula $\psi_\epsilon$ of rank at most $k$
  such that $\supnorm{\phi-\psi_\epsilon}\le\epsilon$ on the final
  coalgebra $\mathcal{F}$. By Lemma~\ref{lem:uniform-approx}, this
  approximation is uniform. 
\end{proof}

\begin{rem}\label{rem:rosen}
  Although it is easy to adapt the unravelling construction to
  preserve finite models by using partial unravelling up to the
  locality depth, this will still not yield a Rosen version of the
  above theorem, i.e.\ one where the semantics is restricted to finite
  models. The reason is that the proof as given above involves the
  final probabilistic transition system, which is infinite. We thus
  leave the proof (or refutation) of such a finite-model version of
  the theorem as an open problem.
\end{rem}

\begin{rem}\label{rem:subdistributions}
  As mentioned in Section~\ref{sec:logics}, a version of the
  characterization theorem for unrestricted subdistributions (i.e.
  where the possible models are coalgebras for the functor
  $[0,1]^\At\times\dfun(1+\id)$) can be recovered with some technical
  adaptations. This mostly concerns the Wasserstein-based distance as
  well as the bisimulation game, as the notion of couplings
  (Definition~\ref{def:coupling}) needs to be changed. A coupling of
  two subdistributions $\pi_1$ on $A$ and $\pi_2$ on $B$ is a
  probability distribution on $(1+A)\times(1+B)$.
  The Wasserstein distance of $\pi_1$ and $\pi_2$ for some
  pseudometric $d$ is then defined as $d^\downarrow(\pi_a,\pi_b) =
  \bigwedge_{\mu\in\pi_a\otimes\pi_b}\int\bar d\dif\mu$, using the
  construction $\bar d$ from Remark~\ref{rem:metric-termination}. 
  As for the changes in the game, $D$ now needs to pick a coupling of
  subdistributions as just defined and when distributing the deviation
  $\epsilon$ over the successor pairs, she needs to pick
  $\epsilon'\colon(1+A)\times(1+B)\to[0,1]$ with the restriction that
  $\epsilon'(a',\ast) = \epsilon'(\ast,b') = 1$ for all $a'\in A$ and
  $b'\in B$.
\end{rem}

\section{Conclusions}

\noindent We established a modal characterization result for
quantitative probabilistic modal logic, which states that every
formula of quantitative probabilistic first-order logic that is
\emph{bisimulation-invariant}, i.e.\ non-expansive w.r.t.\ a natural
notion of behavioural distance on probabilistic transition systems,
can be approximated by modal formulas of bounded modal rank, the bound
being exponential in the rank of the original formula. As discussed in
the introduction, the bound on the modal rank is the crucial feature
of this result. Put differently, on bisimulation-invariant properties,
quantitative probabilistic modal logic is as expressive as
quantitative probabilistic first-order logic, up to approximation in
bounded rank.

We leave several obvious open problems, the most prominent one being
whether our main result can be sharpened to state actual equivalence
of a given bisimulation-invariant first-order formula to a modal
formula rather than only approximability. (Wild et al.\ leave a
similar open problem for the case of fuzzy modal
logic~\cite{WildEA18}.) Moreover, we have already mentioned in
Remark~\ref{rem:rosen} that the version of our main result that
restricts the semantics to finite models, in analogy to Rosen's
finite-model version of van Benthem's theorem~\cite{Rosen97}, remains
open. Further directions for future research include lifting our
methods and results to a coalgebraic level of generality building on
existing work on coalgebraic behavioural
pseudometrics~\cite{km:bisim-games-logics-metric} (for the
quantitative setting; the crisp case has already been
established~\cite{SchroderPattinson10b,LitakEA12,SchroderEA17}), thus
covering, e.g., semiring weighted systems or weighted alternating-time
logics; a treatment of \L{}ukasiewicz semantics of the propositional
connectives; and a characterization theorem for the probabilistic
$\mu$-calculus providing a quantitative version of the
Janin-Wa\l{}ukiewicz theorem~\cite{JaninWalukiewicz95}, which would
characterize the probabilistic $\mu$-calculus within a suitable
quantitative probabilistic monadic second-order logic.


\bibliography{coalgml}



\end{document}